\documentclass[review,onefignum,onetabnum]{siamonline220329}

\usepackage{lipsum}
\usepackage{amsfonts}
\usepackage{graphicx}
\usepackage{epstopdf}
\usepackage{algorithmic}
\ifpdf
  \DeclareGraphicsExtensions{.eps,.pdf,.png,.jpg}
\else
  \DeclareGraphicsExtensions{.eps}
\fi

\usepackage{enumitem}
\setlist[enumerate]{leftmargin=.5in}
\setlist[itemize]{leftmargin=.5in}

\usepackage{amsmath}
\numberwithin{equation}{section}
\newcommand{\bea}{\begin{eqnarray}}
\newcommand{\eea}{\end{eqnarray}}
\newcommand{\beaa}{\begin{eqnarray*}}
\newcommand{\eeaa}{\end{eqnarray*}}

\def\cN{{\mathcal N}}

\def\hC{\mathbb{C}}

\def\hE{\mathbb{E}}

\def\hR{\mathbb{R}}

\def\qq{\qquad}

\def\sig{\mathsf{Sig}^m}

\def\cU{\mathcal{U}}
\def\cZ{\mathcal{Z}}
\newcommand{\ts}{\mathsf{T}}

\newsiamremark{remark}{Remark}
\newsiamremark{hypothesis}{Hypothesis}
\crefname{hypothesis}{Hypothesis}{Hypotheses}
\newsiamthm{claim}{Claim}
\newsiamthm{assumption}{Assumption}

\headers{Deep Signature Algorithm for Multi-dimensional  Path-Dependent   Options}{Erhan Bayraktar, Qi Feng, Zhaoyu Zhang}

\title{Deep Signature Algorithm for Multi-dimensional  Path-Dependent   Options\thanks{Submitted to the editors DATE.
\funding{E. Bayraktar is partially supported by the National Science Foundation under grant DMS-2106556 and by the
Susan M. Smith chair. Q. Feng is partially supported by the National Science Foundation under grant DMS-2306769.}}}

\author{Erhan Bayraktar\thanks{Department of Mathematics, University of Michigan, Ann Arbor, MI, 48109; Email: \url{erhan@umich.edu}.}
\and Qi Feng\thanks{Department of
Mathematics, Florida State University, Tallahassee, FL, 32306; Email: \url{qfeng2@fsu.edu}.}
\and Zhaoyu Zhang\thanks{Department of Mathematics, University of California, Los Angeles, Los Angeles, CA 90095. Email: \url{zhaoyu@math.ucla.edu}.}}

\usepackage{amsopn}

\ifpdf
\hypersetup{
  pdftitle={Deep Signature Algorithm for Multi-dimensional  Path-Dependent   Options},
  pdfauthor={E. Bayraktar, Q. Feng, and Z. Zhang}
}
\fi

\begin{document}
\nolinenumbers
\maketitle

\begin{abstract}
In this work, we study the deep signature algorithms for path-dependent options. We extend the backward scheme in [Hur\'e-Pham-Warin. Mathematics of Computation 89, no. 324 (2020)] for state-dependent FBSDEs with reflections to path-dependent FBSDEs with reflections, by adding the signature layer to the backward scheme. Our algorithm applies to both European and American type option pricing problems while the payoff function depends on the whole paths of the underlying forward stock process. We prove the convergence analysis of our numerical algorithm with explicit dependence on the truncation order of the signature and the neural network approximation errors. Numerical examples for the algorithm are provided  including: Amerasian option under the Black-Scholes model, American option with a path-dependent geometric mean payoff function, and the Shiryaev's optimal stopping problem.
\end{abstract}



\section{Introduction}
The recent development of deep learning and neural networks in optimal stopping \cite{bayer2021optimal, becker2019, becker2021solving, gao2022convergence, gonon2022deep, herrera2021optimal, BernardLelong2021}, optimal control \cite{wang2022deep}, and optimal switching \cite{bayraktar2022neural} promote the developments and application of forward backward stochastic differential equations (FBSDEs) in various financial and economic fields. The idea of combining BSDEs and deep learning is first introduced in \cite{han2018solving}. There are several works further advance such an idea for coupled FBSDEs \cite{han2020convergence}, path-dependent FBSDEs \cite{feng2021deep}, which all relies on the forward Euler scheme of the backward process. On the other hand, the backward scheme of the backward process was recently studied in \cite{wang2018deep, hure2020deep}, which is then further developed to solve optimal stopping \cite{bayer2021optimal, bayer2022pricing}, and optimal switching \cite{bayraktar2022neural} problems.
 
In the current work, we focus on solving path-dependent European and American option pricing problems by using FBSDEs with/without reflection and adding a signature layer. In particular, our algorithm can be used to solve optimal stopping problems in the backward schemes in a path-dependent setting. The path-dependent American option pricing problems have been studied in \cite{dingpricing, auster2022jdoi} using diffusion operator integral methods, \cite{lelong2019} by using Wiener chaos expansion, \cite{hansen2000analytical} by deriving  analytical formulas, \cite{barraquand1996pricing} using forward shooting grid, to list a few. {More precisely, the signature layer and the LSTM neural networks have been used in path-dependent FBSDE algorithm \cite{feng2021deep} and backward PDE algorithm \cite{Sabate2020} for European options. In the current backward setting, we need to modify the signature layer input, as the LSTM structure does not fit into the backward scheme for American options.}  Furthermore, we reduce the regularity assumption for the coefficients of the BSDE, which improves the signature layer idea introduced in \cite{feng2021deep}. We do not need to apply Taylor expansion to apply the signature approximation, which makes our algorithm applicable for more general terminal payoff function and the generator of the BSDE. In particular, we provide the convergence/error analysis regarding the truncation order of the signature layer in our algorithm, which is not known in \cite{feng2021deep, Sabate2020}. Indeed, the error analysis is new in the literature for path-dependent American option pricing and optimal stopping problems. The idea of applying  signature has been used in \cite{bayer2021optimal} to study randomized stopping time associated with the signature, which is not in the FBSDE framework and is not suitable for American type path-dependent payoff function.  A different path-dependent scenario  is recently studied in \cite{herrera2021optimal}, where they consider non-Markovian stock price and state dependent option payoffs. Some truly path-dependent American/Bermudan type option problems, which can not be transformed into  state-dependent problems,  are recently investigated in \cite{lelong2019} with moving average of the stock price, and \cite{becker2019} with delays. Both of the examples are only investigated in one dimension. In particular, the delay problem investigated in \cite{becker2019, becker2020} deals with large time horizon and does not have convergence analysis.  The high dimension examples are carried out in \cite{BernardLelong2021, becker2021solving} for only state-dependent options.  We investigate both scenarios in the current work in a more general setting, which are in line with our error analysis.  We focus on solving path-dependent American type option pricing problems (i.e., optimal stopping problems) by using signature layers, which naturally solve the difficulty from the path-dependent feature. Comparing to these related works on path-dependent American type options and optimal stopping problems, we are able to compute examples in relatively high dimension and provide the error analysis at the same time.  As direct application of our algorithm, it is reasonable to further apply our backward signature scheme together with obliquely reflected BSDEs to solve path-dependent optimal switching problems  in the energy markets \cite{carmona2008pricing, bayraktar2022neural, benezet2022switching} and other related fields. This will be studied in future works.

 The paper is organized as below. In Section 2, we introduce the preliminary and the main algorithm. In Section 3, we show the convergence analysis for the path-dependent backward deep signature FBSDE algorithm with and without reflection. Under the additional smoothness assumption of the vector fields for the forward SDE, we show the error analysis due to the truncation order of the signature process. In Section 4, we focus on numerical examples involving path-dependent American type options: American type Asian option pricing problems, American option with a path-dependent geometric mean payoff function, and the Shiryaev's optimal stopping problem.

 \section{Main algorithm}\label{sec-main}
In this paper we study backward numerical algorithms for the following  path-dependent decoupled FBSDE, for any $0\le t\le T$,
\begin{equation}
\label{Path FBSDEs non reflection}
\begin{cases}
X_{t} = x + \int_{0}^{t}b(s,X_{s})ds + \int_{0}^{t}\sigma(s,X_{s})dW_{s}, \\
Y_{t}= g(X_{\cdot\wedge T}) + \int_{t}^{T}f(s, X_{\cdot\wedge s}, Y_{s}, Z_{s})ds - \int_{t}^{T}Z_{s}dW_{s},\quad Y_T=g(X_{\cdot\wedge T}),
\end{cases}
\end{equation}
and the following path-dependent decoupled FBSDE with reflections, 
\begin{equation}
\label{Path FBSDEs}
\begin{cases}
X_{t} = x + \int_{0}^{t}b(s,X_{s})ds + \int_{0}^{t}\sigma(s,X_{s})dW_{s}, \\
Y_{t}= g(X_{\cdot\wedge T}) + \int_{t}^{T}f(s, X_{\cdot\wedge s}, Y_{s}, Z_{s})ds - \int_{t}^{T}Z_{s}dW_{s}+K_T-K_t,\quad Y_T=g(X_{\cdot\wedge T}),\\
Y_t\ge g(t,X_{\cdot\wedge t}),\quad 0\le t \le T,
\end{cases}
\end{equation}
where $K$ is an adapted non-decreasing process satisfying 
\bea
\int_0^T (Y_t-g(t,X_{\cdot\wedge t})) dK_t=0,
\eea 
and $W$ is a $d_1$-dimensional Brownian motion on a complete filtered probability space $(\Omega, \mathcal F, \mathbb F,$ $\mathbb P)$,  with natural filtration $\mathbb F=(\mathcal F_t)_{0\le t\le T}$. Throughout this paper, we denote $X_{\cdot\wedge t}$ as the whole path of $X$ on $[0,t]$, for $t\in[0,T]$. The path-dependent FBSDE could be used to solve path-dependent European option pricing, which has been studied in a forward Euler scheme with signature layer in \cite{feng2021deep}. In the current paper, we apply the signature layer in the backward Euler scheme and provide explicit error regarding the truncation order of the signature process. The reflected FBSDE \eqref{Path FBSDEs} is motivated by optimal stopping and American option pricing problems (e.g., \cite{longstaff2001valuing, el1997reflected}). The state-dependent and mean field version of \eqref{Path FBSDEs} have been studied in \cite{ma2005representations, li2014reflected, hu2022mean} and the references therein. Following the idea from \cite{hure2020deep} for state-dependent reflected FBSDEs, we propose the following backward scheme in the current path-dependent setting.  Firstly, we consider the following Euler scheme for the forward process $X_t$,
\bea \label{euler for X}
X^n_{t_{i+1}}=X^n_{t_i}+b(t_i,X^n_{t_i})\Delta t_{i}+\sigma(t_i,X^n_{t_i})\Delta W_{t_i},
\eea 
for $i=0,1,\cdots,n-1$, $\Delta t_i=t_{i+1}-t_i=T/n$ and $X_0=x_0$. We then apply the following signature layer to the discrete sequential points $\{X^n_{t_i}\}_{i=0}^n$.  {For a more comprehensive review of signature in the rough paths theory, we refer interesting readers to \cite{lyonsqian, friz2010, frizhairer}.}
\begin{definition}[Signature]
For a bounded variation path $x_t\in\hR^d$, for $t\in[0,T]$,  the  signature of enhanced path $(t,x_t)_{t\in[0,T]}$ is defined as the iterated integrals of $(t,x)_{t\in[0,T]}$. More precisely, for a word $J=(j_1,\cdots,j_k)\in\{1,\cdots,d\}^k$ with size $|J|=k$, we denote $\pi_m$ as the truncation of the signature up to degree $m$ as below. For $x^0_t=t$,we have
\begin{equation}
\pi_m(\mathsf{Sig}(x)_t)=\Big(1,\sum_{j=0}^d \int_0^tdx^j_{t_1},\cdots,\sum_{|J|=m}\int_{0<t_1<\cdots<t_m<t}dx^{j_1}_{t_1}\cdots dx^{j_m}_{t_m}\Big).
\end{equation}
\end{definition}
\begin{definition} Consider a discrete $d$-dimensional time series $(x_{t_{i}})_{i = 1}^{n}$ over time interval $[0, T]$. A signature layer of degree $m$ is a mapping from $\hR^{(d+1) \times n}$ to $\hR^{\widehat{d}}$, which computes {$\sig(x)_{t_i}$}  as an output for any $(t, x)$, where {$\sig(x)_{t_i}$}  is the truncated signature of $(t,x)$ over time interval $[0, t_{i}]$ of degree $m$ as follows:
\bea
\sig(x)_{t_i} = \pi_m(\mathsf{Sig}(x)_{[0,t_{i}]}),
\eea
where $i \in \{1, ..., n\}$ and $\widehat{d}$ is the dimension of the truncated signature.
\end{definition}

Following the above definition, we denote  $\tilde n=n/k$ as the number of segments (i.e. signature layers) for the sequence $\{X^n_{t_i}\}_{i=0}^n$ and denote $k\in \mathbb N^+$ as the number of data points in each segment.  That said, we have $u_i=T/n*i*k$ for $i=1,\cdots, \tilde n.$ We then define 
\bea\label{scheme X}
X^n_{u_{i+1}}=X^n_{u_i}+b(u_i,X^n_{u_i})\Delta u_{i}+\sigma(u_i,X^n_{u_i})\Delta W_{{u_i}},
\eea 
for $i=0,1,\cdots,\tilde n-1$. We then consider the linear interpolation of the sequence $\{X^n_{t_i}\}_{i=0}^{n}$  as the continuous paths input for generating the signature. We are now ready to introduce the following backward scheme, 
\beaa 
\cU_{u_{i+1}}^{\theta,\sig}= F^{\tilde n,\sig}(u_i, X^n_{\cdot\wedge u_i},\cU^{\theta,\sig}_{u_i},\cZ^{\theta,\sig}_{u_i},\Delta u_i,\Delta W_{u_i}),
\eeaa 

where
\begin{equation}\label{scheme Y}
	\begin{split}
&F^{\tilde n,\sig}(u_i,X^n_{\cdot\wedge u_i},\cU^{\theta,\sig}_{u_i},\cZ^{\theta,\sig}_{u_i},\Delta u_i,\Delta W_{u_i})\nonumber \\
&:=\cU^{\theta,\sig}_{u_i}-f(u_i,X^n_{\cdot\wedge u_i}, \cU^{\theta,\sig}_{u_i},\cZ^{\theta,\sig}_{u_i})\Delta u_i+\cZ^{\theta,\sig}_{u_i} \Delta W_{u_i},
	\end{split}
\end{equation} 
and we denote 
\bea \label{NN Y}
\cU_{u_i}^{\theta,\sig}:=\mathcal R^{\theta}(\sig(X^n)_{u_i};\xi  ),\quad
\cZ_{u_i}^{\theta,\sig}:= \mathcal R^{\theta}(\sig(X^n)_{u_i};\eta). 
\eea 
Here, for $\theta=(\xi,\eta)$, we fix the neural network $\mathcal R^{\theta}$ at each time step $u_i$, for $i=0,\cdots,\tilde n-1$, with input dimension $\widehat  d$ for the truncated signature $\sig(X^n)_{u_i}$, and output dimension $\tilde d =1$ for $\mathcal U_{u_i}$ and $\tilde d = d_1$ for $\mathcal Z_{u_i}$. {In particular, we denote $\xi$ (resp. $\eta$) as the parameter for $\mathcal U_{u_i}$ (resp. $\mathcal Z_{u_i}$), which are assumed to be independent}. 
The backward scheme for path-dependent FBSDE \eqref{Path FBSDEs non reflection} is defined as: 
\begin{equation}\label{alg: state}
\begin{cases}
\textbf{Step 1}: \text{Initialization:}~ \widehat \cU_{{u_{\tilde{n}}}}=g(X_{\cdot\wedge T});\\
\textbf{Step 2}: \text{For}~ i=\tilde n-1,\cdots, 0, \text{given}~ \widehat \cU_{u_{i+1}},  \text{compute the minimizer of the loss function}\\
\begin{cases}
 L_{u_i}(\theta)&:= \hE\Big| \widehat \cU_{u_{i+1}}(X^n_{\cdot\wedge u_{i+1}}) -F^{\tilde n,\sig}(u_i,X^n_{\cdot\wedge u_i},\cU^{\theta,\sig}_{u_i},\cZ^{\theta,\sig}_{u_i},\Delta u_i,\Delta W_{u_i})\Big|^2, \\
&\theta^*\in \arg \min_{\theta}  L_{u_i}(\theta).
 \\
\text{Update:} &  \widehat \cU_{u_i}
  =\cU_{u_i}^{\theta^*,\sig}, \text{and set}~ \widehat{\mathcal Z}_{u_i}=\mathcal Z_{u_i}^{\theta^*,\sig}.
\end{cases}
\end{cases}
\end{equation}
And the backward scheme for path-dependent FBSDE with reflection \eqref{Path FBSDEs} is defined as:
\begin{equation} \label{alg:the_alg}
\begin{cases}
\textbf{Step 1 {: }} \text{Initialization}: \widehat \cU_{{u_{\tilde{n}}}}=g(X_{\cdot\wedge T});\\
\textbf{Step 2 {: }}  \text{For } i=\tilde n-1,\cdots, 0, \text{given}~ \widehat \cU_{u_{i+1}},  \text{compute the minimizer of the loss function} \\
\begin{cases}
\widehat L_{u_i}(\theta)&:= \hE\Big| \widehat \cU_{u_{i+1}}(X^n_{\cdot\wedge u_{i+1}}) -F^{\tilde n,\sig}(u_i,X^n_{\cdot\wedge u_i},\cU^{\theta,\sig}_{u_i},\cZ^{\theta,\sig}_{u_i},\Delta u_i,\Delta W_{u_i})\Big|^2,\\
&\theta^*\in \arg \min_{\theta} \widehat L_{u_i}(\theta).\\
\text{Update}: & \widehat \cU_{u_i}
=\max[\cU_{u_i}^{\theta^*,\sig},g(u_i,X_{\cdot\wedge u_i})],~ \text{and set} ~\widehat{\mathcal Z}_{u_i}=\mathcal Z^{\theta^*,\sig}_{u_i}.
\end{cases}\\
\end{cases}
\end{equation}

\section{Convergence Analysis}
Throughout this section, we keep the following assumption.
\begin{assumption}
\label{main assumption}
Let the following assumptions be in force.
\begin{itemize}
\item $b\in \hC(\hR^+\times\hR^d;\hR^d), \sigma\in \hC(\hR^+\times\hR^d;\hR^{d\times d_1}), f\in \hC(\hR^+\times \hC(\hR_+;\hR^d)\times \hR\times \hR^{ d_1};\hR), g\in \hC(\hR^+\times \hC(\hR_+;\hR^d);\hR)$ are Lipschitz continuous functions with Lipschitz constant $L$ in all variables; $b(\cdot, 0), \sigma(\cdot, 0), f(\cdot, 0, 0, 0)$ and $g(\cdot, 0)$ are bounded.
\item $b$ and $\sigma$ are smooth with all their derivatives bounded.
\item The terminal function $g$ satisfies a linear growth condition.
\end{itemize}
\end{assumption}
In this section, we show the convergence analysis for Algorithm \eqref{alg: state} and Algorithm \eqref{alg:the_alg}. We first prove the convergence analysis for the deep signature backward scheme associated with the FBSDE without reflection \eqref{Path FBSDEs non reflection}. In particular, we show the explicit dependence of the truncation order of the signature process in the convergence analysis, which is given in the following key lemma.

\begin{lemma}\label{lemma: sig error}
Under Assumption \ref{main assumption}. For any $2\le m\in\mathbb N^+$, $\Delta t=T/n$, and $\tilde n=n/k$ as the number of segments for the signature process, we have 
\beaa  
\mathbb E\Big[ 
\|\mathsf{Sig}(\tilde X^n)_T-\sig(\tilde X^n)_T \|^2\Big] \le  C_{T,d_1,b,\sigma}(k\Delta t)^{m+1}\tilde n,
\eeaa 
where we denote  $\tilde X^n$ as the continuous interpolation of $X^n$ in \eqref{euler for X}. The constant $C$ depends on $T$, $d_1$, and the bound of $b$, $\sigma$ together with all their derivatives up to order $m$.
\end{lemma}
\begin{proof}
\noindent\textbf{Step 1 (tail estimates):} 
Denote $(X^n_i)_{i=1}^n$ as the discrete approximation from Euler scheme with step size $\Delta t=T/n$, and $\tilde X^n$ as the continuous interpolation of $X^n$.
We consider the following SDE for $(\tilde X^n_t)_{t\in[0,T]}$,
\bea\label{tilde X SDE} 
\tilde X^n_t=X_0+\int_0^t  b(\lfloor s/\Delta t \rfloor \Delta t,\tilde X^n_{\lfloor s/\Delta t \rfloor \Delta t})ds+\int_0^t\sigma(\lfloor s/\Delta t \rfloor \Delta t,\tilde X^n_{\lfloor s/\Delta t \rfloor \Delta t})dW_s,
\eea 
where we have $\tilde X^n_t=\tilde X^n_{i\Delta t}=X^n_{i}$, for $t=i\Delta t$, $i=1,\cdots, n$. The solution $\tilde X^n$ has the following tail probability for general SDEs, (see e.g. \cite{Aze}[Appendix 2] for the Brownian motion case, and a stronger estimates in \cite{Baudoin19}[Proposition 2.10]) for general fractional Brownian motion case),
\bea\label{tail tilde X}
\mathbb P(\sup_{t\in[0,T]}|\tilde X^n_t -x_0|\ge \xi )\le \exp(-\frac{c \xi^{2}}{t}).
\eea 

\noindent\textbf{Step 2 (tail estimates for the remainder term of signature):} The truncated signature $\sig
(\tilde X^n)_t$ of the path $(\tilde X_s^n)_{s\in[0,t]}$ with $t\in[0,T]$  at order m satisfies the following SDE, 
\bea\label{sig m SDE}
d \sig
(\tilde X^n)_t =\sum_{i=1}^{d} V_i( \sig
(\tilde X^n)_t){\circ }d\tilde X^{i,n}_t,
\eea 
and 
\bea 
\label{X_i sde}
\tilde X^{i,n}_t = \int_0^t  b_i(\lfloor s/\Delta t \rfloor \Delta t,\tilde X^n_{\lfloor s/\Delta t \rfloor \Delta t})ds+\int_0^t \sum_{j=1}^{d_1}\sigma_{ij}(\lfloor s/\Delta t \rfloor \Delta t,\tilde X^n_{\lfloor s/\Delta t \rfloor \Delta t})dW^j_s, 
\eea 
where  the vector fields $V_i$ form a horizontal basis of the Carnot group of depth $m$ and the explicit forms of $V_i$, $i=1,\cdots,d$ are given in \cite{friz2010}[Proposition 7.8 \& Remark 7.9]. Note that, \eqref{sig m SDE} is usually taken as a Stratonovich SDE (see e.g. \cite{BaudoinFengOuyang} for truncated signature of fBm) in the rough path sense. In the current It\^o SDE setting, there is an equivalent form (see e.g. \cite{Baudoin}[Remark 1.4]),  we will simply take \eqref{sig m SDE} in the {Stratonovitch} form. This is based on the following equivalent formulation. There exits $\tilde b$ and $\tilde \sigma$, such that, by using the first order projection $\pi_1(\sig(\tilde X)_t)$, we have 
\begin{equation*}
	\begin{split} 
b_i(\lfloor s/\Delta t \rfloor \Delta t,\tilde X^n_{\lfloor s/\Delta t \rfloor \Delta t})&=\tilde  b_i\circ \pi_1(\sig(\tilde X)_t),\\
 \sigma_{ij}(\lfloor s/\Delta t \rfloor \Delta t,\tilde X^n_{\lfloor s/\Delta t \rfloor \Delta t})&=\tilde  \sigma_{ij}\circ \pi_1(\sig(\tilde X)_t).
\end{split}
\end{equation*}
{With some abuse of notation, we use $\circ$ as the composition of two functions in the above definition.} 
Combining \eqref{sig m SDE}, \eqref{X_i sde}  with the above notation, for $B^0_t=t$, we have 
\bea\label{sig m SDE: new}
d \sig
(\tilde X^n)_t =\sum_{j=0}^{d_1} U_j( \sig
(\tilde X^n)_t)\circ dW^j_t,
\eea 
where we have $U_0^i( \sig
(\tilde X^n)_t)= V_i( \sig
(\tilde X^n)_t)\cdot \tilde b_i\circ \pi_1(( \sig
(\tilde X^n)_t))$ and $U_j^i( \sig
(\tilde X^n)_t)=V_i( \sig
(\tilde X^n)_t)\cdot \tilde \sigma_{ij}\circ \pi_1(( \sig
(\tilde X^n)_t))$ for $i=1,\cdots, d$ and $j=1,\cdots, d_1$.
  Following the Chen-Strichartz formula in \cite{Baudoin, Cast, BaudoinFengOuyang} for our SDE \eqref{sig m SDE}, we have 
\bea\label{m chen formula}
\sig(\tilde X)_t=\exp\Big(\sum_{I,l(I)\le m} \Lambda_I(\tilde X^n)_tV_{I} \Big),
\eea  
where we denote $I\in \{1,\cdots, d\}^k$ as a word, and $V_I=[V_{i_1},[V_{i_2},\cdots,[V_{i_{k-1}},V_{i_k}]\cdots ]$ as the iterated Lie bracket. We refer the explicit definition of the exponential map and the explicit formula of $\Lambda_I(\tilde X^n)_t$ to \cite{Baudoin}[Theorem 1.1] and \cite{Cast}[Theorem 2.1]. Indeed, following the Taylor expansion \cite{Aze} and Castell estimates \cite{Cast} for the following SDE associated with the signature process $\mathsf{Sig}
(\tilde X^n)_t$,  
\bea\label{sig SDE}
d \mathsf{Sig}
(\tilde X^n)_t =\sum_{i=1}^{d} V_i(  \mathsf{Sig}
(\tilde X^n)_t) \circ d\tilde X^n_t,
\eea 
we have the following expansion, 
\begin{equation*}
	\begin{split}
\mathsf{Sig}
(\tilde X^n)_t &=\exp\Big(\sum_{I,l(I)\le m} \Lambda_I(\tilde X^n)_tV_{I} \Big)+t^{\frac{m+1}{2}}R_{m+1}(t)\\
&=\sig (\tilde X^n)_t +t^{\frac{m+1}{2}}R_{m+1}(t),
\end{split}
\end{equation*} 
and there exists a random time $\zeta$, some constants $\alpha, c>0$, for every $\xi\ge 1$, such that 
\bea\label{remainder est}
\mathbb P(\sup_{t\in [0,\tau]} \|t^{\frac{m+1}{2}}R_{m+1}(t)\|\ge \xi \tau^{\frac{m+1}{2}}; \tau< \zeta ) \le \exp(-\frac{c\xi^{\alpha} }{\tau}).
\eea 
The above estimates \eqref{remainder est} follows from the tail estimates for SDE \eqref{tail tilde X}. The proof follows from \cite{Aze, Cast} for SDE driven by Brownian motion, see also \cite{fengzhang}[Theorem 3.3 and Remark 3.5 with $H=1/2$].  The random time $\zeta$ is the exit time of $\Lambda_{I, |I|\le m}(\tilde X^n)_t$ in a compact domain $K$. We pick the compact domain $K$ large enough, such that $\Lambda_{I,|I|\le m}(\tilde X^n)_t$ stays inside $K$ for $t\in [0,T]$. 
 
 \noindent \textbf{Step 3 ( remainder estimate on small interval $[0,S]$):} Following \textbf{Step 2}, assuming $t\in[0, S]$ with small $S<T$, and $K$ large enough with $T\le \zeta$, there exist constants $c, \alpha>0$, such that
 \bea
\mathbb P(\sup_{t\in [0,S ]} \|\mathsf{Sig}(\tilde X^n)_t-\sig(\tilde X^n)_t \|\ge \xi S^{\frac{m+1}{2}}) \le \exp(-\frac{c\xi^{\alpha} }{S}).
\eea 
Namely, we have 
  \beaa
\mathbb P( \|\mathsf{Sig}(\tilde X^n)_S-\sig(\tilde X^n)_S \|^2\ge \xi^2 S^{m+1})=\mathbb P( \|\mathsf{Sig}(\tilde X^n)_S-\sig(\tilde X^n)_S \|\ge \xi S^{\frac{m+1}{2}}) \le \exp(-\frac{c\xi^{\alpha} }{S}).
\eeaa 
Denote $Y=\|\mathsf{Sig}(\tilde X^n)_S-\sig(\tilde X^n)_S \|^2$, we have 
\begin{equation}\label{sig L2 bound}
	\begin{split}
\mathbb E[\|\mathsf{Sig}(\tilde X^n)_S-\sig(\tilde X^n)_S \|^2]=&\mathbb E[Y]\\
=&\int_0^{\infty} \mathbb P(Y>y)dy\\
(\text{change of variable:}\quad y=\xi S^{m+1})=& S^{m+1}\int_0^{\infty} \mathbb P(Y\ge \xi S^{m+1}) d\xi\\
\le & S^{m+1} \int_0^{\infty} \exp(-\frac{c \xi^{\alpha}}{S})d\xi \\
\le & C S^{m+1}.
\end{split}
\end{equation}

  \noindent \textbf{Step 4 (Chen's relation ):}
  Denote the time interval as $[0,T]=[0=u_0,u_1]\cup [u_1,u_2]\cup\cdots\cup [u_{\tilde n-1},u_{\tilde n}=T]$.
  By using the Chen's relation \cite{friz2010}[Theorem 7.11] , we have 
  \beaa 
  \mathsf{Sig}(\tilde X^n)_T=\mathsf{Sig}(\tilde X^n)_{[0,u_1]}\otimes \mathsf{Sig}(\tilde X^n)_{[u_1,u_2]}\cdots \mathsf{Sig}(\tilde X^n)_{[u_{\tilde n-1},u_{\tilde n}]},
  \eeaa 
  which holds true for $\sig(\tilde X^n)_T$ as well.
 We then have the following interpolation,
 \begin{equation*}
 	\begin{split}
 &\mathbb E\Big[  \Big\|\mathsf{Sig}(\tilde X^n)_{[0,u_1]}\otimes \mathsf{Sig}(\tilde X^n)_{[u_1,u_2]}\cdots \mathsf{Sig}(\tilde X^n)_{[u_{\tilde n-1},T]}\\
 &\quad\quad- \sig(\tilde X^n)_{[0,u_1]}\otimes \sig(\tilde X^n)_{[u_1,u_2]}\cdots \sig(\tilde X^n)_{[u_{\tilde n-1},T]}\Big\|^2\Big]\\
 =& \mathbb E\Big[ \Big\|\sum_{k=0}^{\tilde n-1} \Big( \otimes_{j=0}^{k}\mathsf{Sig}(\tilde X^n)_{[u_{j},u_{j+1}]}\otimes_{\tilde j=k+1}^{\tilde n-1}\sig(\tilde X^n)_{[u_{\tilde j},u_{\tilde j+1}]}\\
&\quad\quad -\otimes_{j=0}^{k-1}\mathsf{Sig}(\tilde X^n)_{[u_{j},u_{j+1}]}\otimes_{\tilde j=k}^{i-1}\sig(\tilde X^n)_{[u_{\tilde j},u_{\tilde j+1}]}  \Big)\Big\|^2\Big]\\
\le&\mathbb E \Big[\sum_{k=0}^{\tilde n-1}\|\mathsf{Sig}(\tilde X^n)_{[u_k,u_{k+1}]}- \sig(\tilde X^n)_{u_k,u_{k+1}} \|^2 \|\mathsf{Sig}(\tilde X^n)_{[0,T]} \|^2 \Big] \\
 \le &C (\frac{T}{\tilde n})^{m+1}*\tilde n =C (k\Delta t)^{m+1}*\tilde n,
 \end{split}
 \end{equation*} 
 where the last inequality follows from Cauchy-Schwarz inequality {together with applying the estimates in \eqref{sig L2 bound} for each  $\mathbb E \Big[\|\mathsf{Sig}(\tilde X^n)_{[u_k,u_{k+1}]}- \sig(\tilde X^n)_{u_k,u_{k+1}} \|^4\Big] $} and the constant $C$ depends on the upper bound of $\mathbb E[\|\mathsf{Sig}(\tilde X^n)_{[0,T]}\|^4 ]$. Based on the equivalent formulation \eqref{sig m SDE: new}, the dynamics of $\mathsf{Sig}(\tilde X^n)_t$ follows similar SDE driven by Brownian motion with vector fields $U_i$, $i=0,1,\cdots, d_1$. Following \cite{Baudoin}[Proposition 1.3 \& Remark 1.9], we have $\mathbb E[\mathsf{Sig}(\tilde X^n)_t]=\exp(tU_0+ \frac{1}{2}t\sum_{i=1}^{d_1}U_j^2)$. Similar to \eqref{m chen formula}, the signature $\mathsf{Sig}(\tilde X^n)_t$ also has an exponential representation. Given that $\tilde X^n$ and $\Lambda(\tilde X^n)_t$ stays inside a compact set, there exists an upper bound C depends on the time $T$, the dimension $d_1$, and the uniform upper bound of the vector fields $U_i$ together with their derivatives up to order $m$. The proof is thus completed.   
 \qed 
\end{proof}
 
Next, we introduce the following error term $\varepsilon^Z$  as the $L^2$-regularity of $Z$,
\beaa
\varepsilon^Z=:\hE[\sum_{i=0}^{\tilde n-1}\int_{u_i}^{u_{i+1}}|Z_t-\bar Z_{u_i}| dt ],\quad \text{with}\quad \bar Z_{u_i}:= \frac{1}{\Delta u_i}\hE_{u_i}[\int_{u_i}^{u_{i+1}}Z_t dt],
\eeaa
which means that $\varepsilon^Z=\mathcal O(k\Delta t)$ if $g$ is Lipschitz (see e.g. \cite{zhang2004numerical}). Furthermore, we introduce the following error between a continuous function and its neural network approximation at each time step $u_i$, for $i=0, 1\cdots, \tilde n-1.$
\begin{lemma}\label{lemma: error interpolation} 
For any continuous function $\mathsf v_{u_i}\in \hC(\hC([0,u_i];\hR^d);\hR)$ with input $X^n_{\cdot\wedge u_i}$ and its neural network approximation $\cU_{u_i}(\sig(X^n)_{u_i};\xi)$ with input $\sig(X^n)_{u_i}$ and neural network parameter $\xi$, we define the approximation error as below,
\begin{equation}\label{defn: NN error general}
\widetilde \varepsilon_{u_i}^{\cN,\mathsf v}:=\inf _{\xi}\hE|\mathsf v_{u_i}(X^n_{\cdot\wedge u_i}) -\cU_{u_i}(\sig(X^n)_{u_i};\xi)|^2.
\end{equation} 
We then have the following bound of the error,
\bea
\widetilde \varepsilon_{u_i}^{\cN,\mathsf v}\le \varepsilon_{u_i}^{\cN,\mathsf v}+\varepsilon_{u_i}^{\cN,\mathsf v,\mathsf{Sig}, m},
\eea 
where $ \varepsilon_{u_i}^{\cN,\mathsf v}$ denotes the error introduced from the neural network at time $u_i$, and $\varepsilon_{u_i}^{\cN,\mathsf v,\mathsf{Sig}, m}$ denotes the error from the truncation of the signature at order $m$.
\end{lemma}
\begin{proof} 
We apply the following interpolation inequality for the right hand side of \eqref{defn: NN error general},  
\bea\label{interpolation}
&&\mathbb E|\mathsf v_{u_i}(X^n_{\cdot\wedge u_i})-\cU_{u_i}(\pi_m(\sig(X^n)_{u_i};\xi)|^2 \\
&\le&\mathbb E|\mathsf v_{u_i}(X^n_{\cdot\wedge u_i})-\mathcal L(\mathsf{Sig}(X^n)_{u_i})|^2\cdots\cdots \mathcal I \nonumber \\
&&+\mathbb E|\mathcal L(\mathsf{Sig}(X^n)_{u_i})-\mathcal L(\sig(X^n)_{u_i})|^2\cdots\cdots \mathcal{J}\nonumber\\
&&+\mathbb E|\mathcal L(\sig(X^n)_{u_i})-\cU (\sig(X^n)_{u_i};\xi) |^2\cdots\cdots \mathcal{K}\nonumber,
\eea 
where we denote $\mathcal L$ as a linear functional for $\mathsf v_{u_i}$. According to the universal nonlinearity proposition [e.g., \cite{arribas2018derivatives}, see also \cite{kidger2019deep} Proposition A.6], for continuous function $\mathsf v_{u_i}$, and continuous paths $X_{\cdot \wedge u_i}$, for any $\varepsilon_{u_i}^{\cN,\mathsf v,1}>0$,  there exists a linear functional $\mathcal L$ depending on a compact set $\mathfrak K$ (i.e. a compact subset of the space of all geometric p-rough paths defined on $[0,T]$, see e.g. \cite{kalsi2020}[Lemma 3.4]), such that 
\beaa 
\mathcal I=\mathbb E|\mathsf v_{u_i}(X^n_{\cdot\wedge u_i})-\mathcal L(\mathsf{Sig}(X^n)_{u_i})|^2\le \varepsilon_{u_i}^{\cN,\mathsf v,1}.
\eeaa 
Since $\mathcal L$ is a linear functional, applying the remainder term estimates for the signature of the path (e.g. continuous interpolation) generated by $\{X^n_{u_1},\cdots,X^n_{u_{\tilde n}}\}$, we have
\beaa
\mathcal J&=&\mathbb E|\mathcal L(\mathsf{Sig}(X^n)_{u_i})-\mathcal L(\sig(X^n)_{u_i})|^2
\le C_{T,d_1,b,\sigma} (k\Delta t)^{m+1}i:= \varepsilon_{u_i}^{\cN,\mathsf v,\mathsf{Sig},m},
\eeaa 
for $i=1,\cdots, \tilde n$, where  the error   comes from Lemma \ref{lemma: sig error} on $[u_0,u_i]$. In particular, we denote $C_{T,d_1,b,\sigma}$ as the uniform upper bound for the linear functional $\mathcal L$ among all the time step $u_i$ based on the compact set $\mathfrak K$.  Next, applying the universality property \cite{funahashi1993approximation}, for the linear functional $\mathcal L$, for any $ \varepsilon_{u_i}^{\cN,\mathsf v,2}>0$,  there exists a neural network $\mathcal U$, such that 
\beaa 
\mathcal K=\mathbb E|\mathcal L(\sig(X^n)_{u_i})-\cU (\sig(X^n)_{u_i};\xi) |^2\le \varepsilon_{u_i}^{\cN,\mathsf v,2}.
\eeaa 
Combining the above estimates, for a fixed neural network at time $u_i$, we conclude that 
\beaa
\widetilde \varepsilon_{u_i}^{\cN,\mathsf v}\le \varepsilon_{u_i}^{\cN,\mathsf v, 1}+\varepsilon_{u_i}^{\cN,\mathsf v,2}+\varepsilon_{u_i}^{\cN,\mathsf v,\sig}\le \varepsilon_{u_i}^{\cN,\mathsf v}+\varepsilon_{u_i}^{\cN,\mathsf v,\mathsf{Sig}, m},
\eeaa 
where we denote 
\beaa
\varepsilon_{u_i}^{\cN,\mathsf v}:=\varepsilon_{u_i}^{\cN,\mathsf v,1}+\varepsilon_{u_i}^{\cN,\mathsf v,2}.
\eeaa
 \qed 
\end{proof}

\begin{theorem}
Under Assumption \ref{main assumption}, for Algorithm \eqref{alg: state} and FBSDE \eqref{Path FBSDEs non reflection}, there exists a constant $C$ depending on $T$ and Lipschitz constant $L$, such that, 
\begin{equation}
	\begin{split}
&\max_{1\le i\le \tilde n-1} \hE[|Y_{u_i}-\widehat{ \mathcal U}_{u_i}|   ]+\sum_{i=0}^{\tilde n-1}\hE[\int_{u_i}^{u_{i+1}}|Z_t-\widehat{\mathcal Z}_{u_i}|^2dt ]\\
\le &  C\hE|g(X_{\cdot\wedge T})-g(X^n_{\cdot\wedge T})|^2+Ck\Delta t+C\varepsilon^Z+\varepsilon^{\mathsf{Sig},m}+C\sum_{i=0}^{\tilde n-1}(\tilde n \varepsilon_{u_i}^{\cN,\mathsf v}+\varepsilon_{u_i}^{\cN,\mathsf z}).
\end{split}
\end{equation}
{ Here $\varepsilon_{u_i}^{\cN,\mathsf v}$ and $\varepsilon_{u_i}^{\cN,\mathsf z}$ denote the errors at each time step $u_i$ as defined in Lemma \ref{lemma: error interpolation}, and $\varepsilon^{\mathsf{Sig},m}$ denotes the accumulated error from the truncation of the signature at order $m$, with 
$\varepsilon^{\mathsf{Sig},m}\le C_{T,d_1,b,\sigma} \tilde n^4(k\Delta t)^{m+1}$ and $k\Delta t=T/\tilde n$.}
\end{theorem}
\begin{remark}
We refer to \cite{schmit2020} for the error analysis introduced from the neural networks in the above theorem.
\end{remark}
\begin{proof}
Following \cite[Theorem 4.1]{hure2020deep} for the state-dependent backward implicit scheme, we consider the backward scheme in the current path-dependent setting as below,  
\beaa 
Y_{u_{i}}=\hE_{u_i}[Y_{u_{i+1}} ]+\hE_{u_i}\left[\int_{u_i}^{u_{i+1}} f(t,X_{\cdot\wedge t},Y_t,Z_t)dt\right].
\eeaa 
Furthermore, we could define the implicit scheme, with $k\Delta t_i=\Delta u_i=u_{i+1}-u_i$ and $\Delta W_{u_i}=W_{u_{i+1}}-W_{u_i}$, 
\begin{equation*}
	\begin{split}
\widehat{ \mathcal V }_{u_i}&= \hE_{u_i}[\widehat{\mathcal U}_{u_{i+1}}(\sig(X^n)_{ u_{i+1}})]+f(u_i,X^n_{\cdot\wedge u_i},\widehat{\mathcal V}_{u_i},\overline{\widehat Z}_{u_i}) k\Delta t,\\
{\overline{\widehat Z}_{u_i} }&=\frac{1}{k\Delta t}\hE_{u_i}[\widehat{\mathcal U}_{u_{i+1}}(\sig(X^n)_{u_{i+1}})\Delta W_{u_i} ].
\end{split}
\end{equation*} 
We investigate the convergence analysis for the following quantity, 
\beaa 
\max_{i=0,1,\cdots,\tilde n} \hE|Y_{u_i}-\widehat\cU_{u_i}(\sig(X^n)_{ u_i}) |^2+\hE\Big[\sum_{i=0}^{\tilde n-1}\int_{u_i}^{u_{i+1}}|Z_{t}-\widehat \cZ_{u_i}(\sig(X^n)_{u_i})|^2 {dt} \Big],
\eeaa 
where $\widehat{\mathcal U}_{u_i}(\sig(X^n)_{u_i}) =\cU_{u_i}(\sig(X^n)_{u_i};\theta^*_{\mathcal U})$ and $\widehat\cZ(\sig(X^n)_{u_i})=\cZ(\sig(X^n)_{u_i};\theta^*_{\mathcal Z})$. Due to the Markovian property of the discretized forward process $X^n_{t_i}$, there exists deterministic continuous functions $\mathsf v_{u_i},\mathsf z_{u_i}: \hC([0,T],\hR^d)\rightarrow \hR$  such that,
\bea \label{deterministic function}
\widehat{\mathcal V}_{u_i}=\mathsf v_{u_i}(X^n_{\cdot\wedge u_i}),\quad \overline{\widehat Z}_{u_i}=\mathsf z_{u_i}(X^n_{\cdot\wedge u_i}).
\eea 
By the martingale representation theorem, there exists an $\mathbb R^d$-valued square integrable process $\widehat Z_t$, such that 
\begin{equation}\label{martingale rep} 
	\begin{split}
\widehat{\mathcal U}_{u_{i+1}}(\sig(X^n)_{u_{i+1}}):=&\widetilde{\widehat{\mathcal U}}_{u_{i+1}}(X^n_{\cdot\wedge u_{i+1}})\\
=&\widehat{\mathcal V}_{u_i}-f(u_i,X^n_{\cdot \wedge u_i}, \widehat{\mathcal V}_{u_i}, \overline{\widehat Z}_{u_i})\Delta u_i+\int_{u_i}^{u_{i+1}} \widehat Z^{\ts}dW_s,
	\end{split}
\end{equation}
and we have,
\beaa 
\overline{\widehat Z}_{u_i}=\frac{1}{k\Delta t}\hE_{u_i}[\int_{u_i}^{u_{i+1}}\widehat Z_s ds ],~\text{for}~~ i=1,\cdots, \tilde n.
\eeaa 
According to the implicit scheme above, we have 
\begin{equation}
	\begin{split}
Y_{u_i}-\widehat{\mathcal V}_{u_i}
=&\hE_{u_i}[Y_{u_{i+1}}-\widehat{\mathcal U}_{u_{i+1}}(\sig(X^n)_{ u_{i+1}}) ]\\
&+\hE_{u_i}[\int_{u_i}^{u_{i+1}} f(t,X_{\cdot\wedge t},Y_t,Z_t)-f(u_i,X^n_{\cdot\wedge u_i},\widehat{\mathcal V}_{u_i},\overline{\widehat Z}^n_{u_i}) {dt} ].
\end{split}
\end{equation}
By using the Young inequality, Cauchy-Schwartz inequality, and Assumption \ref{main assumption}, we get
\bea\label{first inq}
&&\hE |Y_{u_i}-\widehat{\mathcal V}_{u_i}|^2
\le (1+\gamma \Delta u_i) \hE \Big|[Y_{u_{i+1}}-\widehat{ \mathcal{U}}_{u_{i+1}}(\sig(X^n)_{u_{i+1}}) ] \Big|^2 +\frac{4L^2(1+\gamma\Delta u_i)}{\gamma} \nonumber\\ 
&&\quad\qq\times \Big\{(\sup_{t\in [u_i,u_{i+1}]} \hE |X_{\cdot \wedge t}-X^n_{\cdot\wedge u_i}|^2)^2 + 2\Delta u_i \hE |Y_{u_i}-\widehat{\mathcal V}_{u_i} |^2+\hE[\int_{u_i}^{u_{i+1}}|Z_t-\overline{\widehat Z}_{u_i}|^2dt ] \Big\}.
\eea
Here the parameter $\gamma$ is a constant to be chosen later. 
Following \cite{platen2010numerical} [Theorem 9.6.2], for some nonnegative function $c(\Delta u_i)$ with $\lim_{\Delta u_i\downarrow 0} c(\Delta u_i)=0$, we obtain 
\bea \label{path estiamte}
\sup_{t\in [u_i,u_{i+1}]} \hE |X_{\cdot \wedge t}-X^n_{\cdot\wedge u_i}|^2\le C(\Delta u_i+c(\Delta u_i)).
\eea 
 The estimate of $ (\sup_{t\in [u_i,u_{i+1}]} \hE |X_{\cdot \wedge t}-X^n_{\cdot\wedge u_i}|^2)^2$ is the first difference here compared to the state-dependent FBSDE setting. For $\Delta u_i$ small enough, choosing $\gamma=8dL^2$,  following \textit{Step 1} in the proof of \cite{hure2020deep}[Theorem 4.1] for the backward implicit Euler scheme, 
\bea \label{est 412}
&&\hE |Y_{u_i}-\widehat{\mathcal V}_{u_i}|^2\le (1+C| \Delta u_i|) \hE \Big|[Y_{u_{i+1}}-\widehat{ \mathcal{U}}_{u_{i+1}}(\sig(X^n)_{u_{i+1}}) ] \Big|^2+C(\Delta u_i+c(\Delta u_i))^2\nonumber \\
&&\qq\qq\qq+C\hE[\int_{u_i}^{u_{i+1}}|Z_t-\overline{\widehat Z}_{u_i}|^2dt ]+C|\Delta u_i|\hE[\int_{u_i}^{u_{i+1}}|f(t,X_t,Y_t,Z_t)|^2dt ].
\eea 
Applying the Young inequality, i.e. $(a+b)^2\ge (1-|\Delta u_i|)a^2+(1-\frac{1}{|\Delta u_i})b^2\ge (1-|\Delta u_i|)a^2-\frac{1}{|\Delta u_i|}b^2,$ we have 
\begin{equation}\label{est 413}
	\begin{split}
\hE|Y_{u_i}-\widehat{\mathcal V}_{u_i}|^2 =&\hE |Y_{u_i}-\widehat{\mathcal U}_{u_i}(\sig(X^n)_{u_i})+\widehat{\mathcal U}_{u_i}(\sig(X^n)_{u_i})-\widehat{\mathcal V}_{u_i}|^2 \\
\ge & (1-|\Delta u_i|)\hE|Y_{u_i}-\widehat{\mathcal U}_{u_i}(\sig(X^n)_{u_i})|^2-\frac{1}{|\Delta u_i|}\hE|\widehat{\mathcal U}_{u_i}(\sig(X^n)_{u_i})- \widehat{\mathcal V}_{u_i}|^2.\nonumber
	\end{split}
\end{equation}
Plugging  estimate \eqref{est 413} into estimate \eqref{est 412}, for $|\Delta u_i|=k\Delta t$ small enough, we have 
\begin{equation}\label{second inq}
\begin{split}
&\hE|Y_{u_i}-\widehat{\mathcal U}_{u_i}(\sig(X^n)_{u_i}) |^2\nonumber \\
\le& (1+Ck\Delta t) \hE|Y_{u_{i+1}}-\widehat{\mathcal U}_{u_{i+1}}(\sig(X^n)_{u_{i+1}}) |^2+C(k\Delta t+c(k\Delta t))^2\nonumber \\
&+C\hE[\int_{u_i}^{u_{i+1}}|Z_t-\bar Z_{u_i}|^2dt ]+Ck\Delta t\hE[\int_{u_i}^{u_{i+1}}|f(t,X_t,Y_t,Z_t)|^2dt ]\nonumber \\
&+C \tilde n \hE|\widehat{\mathcal V}_{u_i}-\widehat{\mathcal U}_{u_i}(\sig(X^n)_{u_i}) |^2.
\end{split}
\end{equation}
Applying the Gronwall's inequality, and using the error $\varepsilon^Z$ for the $L^2$-regularity of $Z$ and the fact $\hE[\int_0^T f(t,X_{\cdot\wedge t},Y_t,Z_t)^2dt]<\infty$, and recalling the fact $Y_{u_{\tilde n}}=g(X_{\cdot\wedge T})$, and $\widehat{\mathcal U}_{u_{\tilde n}}(\sig(X^n)_{T})=g(X_{\cdot\wedge T})$, we have
\begin{equation}\label{label: estimates 1}
\begin{split}
&\max_{i=0,1,\cdots,\tilde n-1} \hE|Y_{u_i}-\widehat{\mathcal U}_{u_i}(\sig(X^n)_{ u_i}) |\\
&\le C\hE|g(X_{\cdot\wedge T})-g(X^n_{\cdot\wedge T})|^2+Ck\Delta t+C\varepsilon^Z+C\tilde n\sum_{i=0}^{\tilde n-1}\hE|\widehat{\mathcal V}_{u_i}-\widehat{\mathcal U}_{u_i}(\sig(X^n)_{\wedge u_i}) |^2,
\end{split}
\end{equation} 
where we use the estimate $C((k\Delta t)+c(k\Delta t))^2<Ck\Delta t$ when $k\Delta t$ is small enough. Now, we are left to estimate  $C\tilde n\sum_{i=0}^{\tilde n-1}\hE|\widehat{\mathcal V}_{u_i}-\widehat{\mathcal U}_{u_i}(X^n_{\cdot\wedge u_i}) |^2$,  which is the main difference compared to \textit{Step 3} and \textit{Step 4} in \cite{hure2020deep}[Theorem 4.1]. Following \eqref{deterministic function},  we first define the following errors at each time step $u_i$, for $i=1,\cdots,\tilde n$.
\begin{equation}\label{defn: NN error}
	\begin{split}
\widetilde \varepsilon_{u_i}^{\cN,\mathsf v}&:=\inf _{\xi}\hE|\mathsf v_{u_i}(X^n_{\cdot\wedge u_i}) -\cU_{u_i}(\sig(X^n)_{u_i};\xi)|^2; \\
\widetilde  \varepsilon_{u_i}^{\cN,\mathsf z}&:= \inf_{\eta}\hE|\mathsf z_{u_i}(X^n_{\cdot\wedge u_i})-\cZ_{u_i}(\sig(X^n)_{u_i};\eta)|^2.
	\end{split}
\end{equation} 
We then decompose the error in \eqref{alg: state} as $L_{u_i}(\theta):= \widetilde L_{u_i}+\hE[\int_{u_i}^{u_{i+1}}|\widehat Z_t-\overline{\widehat Z} |^2 dt ]$ for $\theta=(\xi,\eta)$, where we define 
\begin{equation}\label{loss}
	\begin{split}
\widetilde L_{u_i}(\theta):=&\Delta u_i\hE|\overline{\widehat Z}_{u_i}-\mathcal Z_{u_i}(\sig(X^n)_{u_i};\eta) |^2\nonumber \\
&+\hE\Big|\widehat{\mathcal V}_{u_i}- {\cU_{u_i}(\sig(X^n)_{u_i};\xi) } +\Big[f(u_i,X_{\cdot\wedge u_i},\mathcal U_{u_i}(\sig(X^n)_{u_i}),\mathcal Z_{u_i}(\sig(X^n)_{u_i}) )\nonumber \\
&\qq\qq\qq\qq -f(u_i,X_{\cdot\wedge u_i},\widehat{\mathcal V}_{u_i},\overline{\widehat Z}_{u_i}) \Big]\Delta u_i \Big|^2.
	\end{split}
\end{equation}
In the following, we show that the errors defined above are indeed small at each time $u_i$.  We first observe that,
\begin{equation}\label{error upper bound}
	\begin{split}
\widetilde L_{u_i}(\theta)&\le (1+C\Delta u_i)\hE|\widehat{\mathcal V}_{u_i}-\mathcal U_{u_i}(\sig(X^n)_{u_i};\xi) |^2\\
&+C\Delta u_i\hE|\overline{\widehat Z}_{u_i}-\mathcal Z_{u_i}(\sig(X^n)_{u_i};\eta) |^2. \\
	\end{split}
\end{equation}
Furthermore, applying the Young inequality: $(a+b)^2\ge (1-\gamma\Delta u_i)a^2+(1-\frac{1}{\gamma\Delta u_i})b^2\ge (1-\gamma\Delta u_i)a^2-\frac{1}{\gamma\Delta u_i}b^2$, we have 
\begin{equation*}
	\begin{split}
\widetilde L_{u_i}(\theta)&\ge\Delta u_i \hE|\overline{\widehat Z}_{u_i}-\mathcal Z_{u_i}(\sig(X^n)_{u_i};\eta) |^2+ (1-\gamma\Delta u_i)\hE|\widehat{\mathcal V}_{u_i}-\mathcal U_{u_i}(\sig(X^n)_{u_i};\xi) |^2\\
&-\frac{2\Delta u_i L^2}{\gamma} \Big(\hE|\widehat{\mathcal V}_{u_i}-\mathcal U_{u_i}(\sig(X^n)_{u_i};\xi) |^2+ \Delta u_i\hE|\overline{\widehat Z}_{u_i}-\mathcal Z_{u_i}(\sig(X^n)_{u_i};\eta) |^2 \Big).
\end{split}
\end{equation*}
Let $\gamma=4L^2$, we get 
\beaa
\widetilde L_{u_i}(\theta)\ge (1-C\Delta u_i)\hE|\widehat{\mathcal V}_{u_i}-\mathcal U_{u_i}(\sig(X^n)_{u_i};\xi) |^2+\frac{\Delta u_i}{2}\hE|\overline{\widehat Z}_{u_i}-\mathcal Z_{u_i}(\sig(X^n)_{u_i};\eta) |^2.
\eeaa
Combining with \eqref{error upper bound}, we observe that 
\begin{equation*}
	\begin{split}
&(1-C\Delta u_i)\hE|\widehat{\mathcal V}_{u_i}-\mathcal U_{u_i}(\sig(X^n)_{u_i};\xi) |^2+\frac{\Delta u_i}{2}\hE|\overline{\widehat Z}_{u_i}-\mathcal Z_{u_i}(\sig(X^n)_{u_i};\eta) |^2\le \widetilde L_{u_i}(\theta^*)\\
\le & \widetilde L_{u_i}(\theta)\le (1+C\Delta u_i)\underbrace{\hE|\widehat{\mathcal V}_{u_i}-\mathcal U_{u_i}(\sig(X^n)_{u_i};\xi) |^2}_{\mathbf I}+C\Delta u_i\underbrace{ \hE|\overline{\widehat Z}_{u_i}-\mathcal Z_{u_i}(\sig(X^n)_{u_i};\eta) |^2}_{\mathbf J}.
	\end{split}
\end{equation*}

{Applying Lemma \ref{lemma: error interpolation} for the errors defined in \eqref{defn: NN error} or equivalently the terms $\mathbf I$ and $\mathbf J$ above},  we conclude that 
\bea\label{new error v}
\widetilde \varepsilon_{u_i}^{\cN,\mathsf v}\le \varepsilon_{u_i}^{\cN,\mathsf v, 1}+\varepsilon_{u_i}^{\cN,\mathsf v,2}+\varepsilon_{u_i}^{\cN,\mathsf v,\sig}\le \varepsilon_{u_i}^{\cN,\mathsf v}+\varepsilon_{u_i}^{\cN,\mathsf v,\mathsf{Sig}, m},
\eea 
where we denote 
\bea
\varepsilon_{u_i}^{\cN,\mathsf v}:=\varepsilon_{u_i}^{\cN,\mathsf v,1}+\varepsilon_{u_i}^{\cN,\mathsf v,2}.
\eea

Similarly, we obtain the same estimates for $\widetilde \varepsilon_{u_i}^{\cN,\mathsf z}$ according to Lemma \ref{defn: NN error general}, namely,
\bea\label{new error z}
\widetilde  \varepsilon_{u_i}^{\cN,\mathsf z}\le \varepsilon_{u_i}^{\cN,\mathsf z, 1}+\varepsilon_{u_i}^{\cN,\mathsf z,2}+\varepsilon_{u_i}^{\cN,\mathsf z,\sig}\le \varepsilon_{u_i}^{\cN,\mathsf z}+\varepsilon_{u_i}^{\cN,\mathsf z,\mathsf{Sig},m}.
\eea 
In particular, for each fixed $u_i$, we apply Lemma \ref{lemma: sig error} to get the upper bounds of $\varepsilon_{u_i}^{\mathcal N, \mathsf v, \mathsf{Sig}, m}$ and $\varepsilon_{u_i}^{\mathcal N, \mathsf z, \mathsf{Sig},m}$. 
From now on, we will refer to 
\bea \label{explicit error}
\varepsilon_{u_i}^{\cN,\mathsf v},\quad \varepsilon_{u_i}^{\cN,\mathsf z},\quad \varepsilon_{u_i}^{\cN,\mathsf v,\mathsf{Sig},m},\quad \varepsilon_{u_i}^{\cN,\mathsf z,\mathsf{Sig},m},
\eea
as the error for \eqref{defn: NN error} from neural network and the signature layer. For $\Delta u_i$ small enough, we end up with
\bea\label{nn error euler}
&&\hE|\widehat{\mathcal V}_{u_i}-\mathcal U_{u_i}(\sig(X^n)_{u_i};\xi) |^2+\Delta u_i\hE|\overline{\widehat Z}_{u_i}-\mathcal Z_{u_i}(\sig(X^n)_{u_i};\eta) |^2\\
&&\le C( \varepsilon_{u_i}^{\mathcal N,\mathsf v}+\varepsilon_{u_i}^{\mathcal N,\mathsf v, \mathsf{Sig},m}+\Delta u_i (\varepsilon_{u_i}^{\mathcal N,\mathsf z}+\varepsilon_{u_i}^{\mathcal N,\mathsf z, \mathsf{Sig},m} ) ).\nonumber  
\eea
 Applying the error defined in \eqref{defn: NN error} and \eqref{explicit error} to \eqref{label: estimates 1}, we have
\begin{equation*}
	\begin{split}
&\max_{i=0,1,\cdots,\tilde n-1} \hE|Y_{u_i}-\widehat{\mathcal U}_{u_i}(\sig(X)_{u_i}) |\\
&\le   C\hE|g(X_{\cdot\wedge T})-g(X_{\cdot\wedge T}^n)|^2+Ck\Delta t+C\varepsilon^Z\\
&+C\tilde n\sum_{i=0}^{\tilde n-1}(\tilde n (\varepsilon_{u_i}^{\cN,\mathsf v}+\varepsilon_{u_i}^{\cN,\mathsf v, \mathsf{Sig},m})+\varepsilon_{u_i}^{\cN,\mathsf z}+\varepsilon_{u_i}^{\cN,\mathsf z, \mathsf{Sig},m})\\
&\le   C\hE|g(X_{\cdot\wedge T})-g(X_{\cdot\wedge T}^n)|^2+Ck\Delta t+C\varepsilon^Z+\varepsilon^{\mathsf{Sig},m}+C\tilde n\sum_{i=0}^{\tilde n-1}(\tilde n \varepsilon_{u_i}^{\cN,\mathsf v}+\varepsilon_{u_i}^{\cN,\mathsf z}),
	\end{split}
\end{equation*} 
 where we denote $\varepsilon^{\mathsf{Sig},m}$ as the total accumulated error from the signature truncation, and it has the following upper bound according to Lemma \ref{lemma: sig error},
\bea\label{cummulated sig error}
\varepsilon^{\mathsf{Sig},m}=C\tilde n\sum_{i=0}^{\tilde n-1}(\tilde n \varepsilon_{u_i}^{\cN,\mathsf v, \mathsf{Sig},m}+\varepsilon_{u_i}^{\cN,\mathsf z, \mathsf{Sig},m})\le C_{T,d_1,b,\sigma} (k\Delta t)^{m+1}\tilde n^4.
\eea
The estimates for $\hE[\int_{u_i}^{u_{i+1}}|Z_t-\widehat{\mathcal Z}_{u_i}|^2dt ]$ follows from the following observation, 
\begin{equation} 
\begin{split}
\hE[\int_{u_i}^{u_{i+1}}|Z_t-\widehat{\mathcal Z}_{u_i}|^2dt ]&\le 2\hE[\int_{u_i}^{u_{i+1}}|Z_t-\overline{\widehat Z}_{u_i}|^2dt ]+2\Delta u_i\hE|\overline{\widehat Z}_{u_i}-\widehat{\mathcal Z}_{u_i}(\sig(X)_{u_i}) |^2\\
&=:\mathbf J_1+\mathbf J_2.
\end{split}
\end{equation} 
The estimate of $\mathbf J_2$ follows from error \eqref{defn: NN error}, the estimate of $\mathbf J_1$ is similar to \textit{Step 5} in \cite{hure2020deep}[Theorem 4.1] by using our neural network error defined in \eqref{defn: NN error},
\begin{equation}
	\begin{split}
\sum_{i=0}^{\tilde n-1}\hE[\int_{u_i}^{u_{i+1}}|Z_t-\overline{\widehat Z}_{u_i}|^2dt ] &\le C\hE|g(X_{\cdot\wedge T})-g(X_{\cdot\wedge T}^n)|^2+Ck\Delta t+C\varepsilon^Z+\varepsilon^{\mathsf{Sig},m}\\&+C\sum_{i=0}^{\tilde n-1}(\tilde n \varepsilon_{u_i}^{\cN,\mathsf v}+\varepsilon_{u_i}^{\cN,\mathsf z}).
	\end{split}
\end{equation}
Combining the above two estimate, we complete the proof.
 \qed 
\end{proof}

\subsection{Convergence Analysis: reflected FBSDE}
We are now ready to present the convergence analysis for Algorithm \ref{alg:the_alg}, which produces the following scheme, 
\begin{equation}
\label{RBSDE-1}
\begin{cases}
\widetilde{\mathcal V}_{u_i}= \hE_{u_i}[\widehat{\mathcal U}_{u_{i+1}}(\sig(X^n)_{ u_{i+1}})]+f(u_i,X^n_{\cdot\wedge u_i},\widetilde{\mathcal V}_{u_i},\overline{\widetilde Z}_{u_i})\Delta u_i,\\
{\overline{\widetilde Z}_{u_i}}=\frac{1}{k\Delta t}\hE_{u_i}[\widehat{\mathcal U}_{u_{i+1}}(\sig(X^n)_{ u_{i+1}})\Delta W_{u_i} ],\\
\widehat{\mathcal V}_{u_i}:=\max[\widetilde{\mathcal V}_{u_i},g(u_i,X^n_{\cdot\wedge u_i}) ]
\end{cases}
\end{equation}
Similar to Lemma \ref{defn: NN error general}, and 
\eqref{defn: NN error}, \eqref{new error v}, \eqref{new error z}, we consider the error below,
\begin{equation}
\begin{split}\label{defn: NN error reflection}
\widetilde \varepsilon_{u_i}^{\cN,\tilde{ \mathsf v}}:=\inf _{\xi}\hE|\tilde{\mathsf v}_{u_i}(X^n_{\cdot\wedge u_i}) -\cU_{u_i}(\sig(X^n)_{u_i};\xi)|^2=\varepsilon^{\mathcal N, \tilde{\mathsf v}}_{u_i}+\varepsilon^{\mathcal N, \tilde{\mathsf v}, \mathsf{Sig},m}_{u_i}; \\
\widetilde \varepsilon_{u_i}^{\cN,\tilde{\mathsf z}}:= \inf_{\eta}\hE|\tilde{\mathsf z}_{u_i}(X^n_{\cdot\wedge u_i})-\cZ_{u_i}(\sig(X^n)_{u_i};\eta)|^2=\varepsilon^{\mathcal N, \tilde{\mathsf z}}_{u_i}+\varepsilon^{\mathcal N, \tilde{\mathsf z}, \mathsf{Sig},m}_{u_i},
\end{split}
\end{equation}
which could be bounded as shown in \eqref{defn: NN error}.

\begin{theorem}
Let Assumption \ref{main assumption} hold. For Algorithm \eqref{alg:the_alg} and reflected FBSDE \eqref{Path FBSDEs}, there exists a constant $C>0$ depending on $T$ and Lipschitz constant $L$, such that, 
\begin{equation}
\begin{split}
&\max_{0\le i\le \tilde n-1} \hE[|Y_{u_i}-\widehat{ \mathcal U}_{u_i}(\sig(X^n)_{u_i})|   ]+\sum_{i=0}^{\tilde n-1}\hE[\int_{u_i}^{u_{i+1}}|Z_t-\widehat{\mathcal Z}_{u_i}(\sig(X^n)_{u_i})|^2dt ]\\
\le&   C(k\Delta t + \sum_{i=0}^{\tilde n-1}(\tilde n \varepsilon_{u_i}^{\cN,\tilde{\mathsf v}}+\varepsilon_{u_i}^{\cN,\tilde{\mathsf z}}))+\varepsilon^{\mathsf{Sig},m}.	
\end{split}
\end{equation}
Here $\varepsilon_{u_i}^{\cN,\tilde{\mathsf v}}$ and $\varepsilon_{u_i}^{\cN,\tilde{\mathsf z}}$ denotes the error introduced in \eqref{defn: NN error reflection} at time $u_i$ from the neural network approximation, and $\varepsilon^{\mathsf{Sig},m}$ denotes the accumulated error from the truncation of the signature at order $m$, with 
$\varepsilon^{\mathsf{Sig},m}\le C_{T,d_1,b,\sigma} \tilde n^4(k\Delta t)^{m+1}$ and $k\Delta t=T/\tilde n$.
\end{theorem}
\begin{proof}
We first introduce the discrete-time approximation of the path-dependent reflected BSDE,
\begin{equation}\label{RBSDE-2}
\begin{cases}
Y^n_{u_{\tilde n}}=g(X^n_{\cdot\wedge u_{\tilde n}}),\\
Z_{u_i}^n=\frac{1}{\Delta u_i}\hE_{u_i}[Y^n_{u_{i+1}}\Delta W_{u_i}]\\
\widetilde Y^n_{u_i}=\hE_{u_i}[Y^n_{u_{i+1}}]+f(u_i,X_{\cdot\wedge u_i},\widetilde Y_{u_i}^n,Z^n_{u_i})\Delta u_i,\\
Y_{u_i}^n=\max[\widetilde Y_{u_i}^n, g(X_{\cdot\wedge u_i})],~~i=0,\cdots,\tilde n-1.
\end{cases}
\end{equation}
We first need the Euler scheme estimate for the path-dependent reflected BSDE.  According to the well-posedness for a general oblique reflected BSDE in \cite{hamadene2010switching}[Section 2.2], in our current situation, $Y_t$ is one-dimensional and the path-dependence of the generator $f$ on $X_{\cdot\wedge t}$ does not affect the measurability assumption in \cite{hamadene2010switching}[Section 2.2], the existence of the RBSDE follows directly.  The rate of convergence  follows similar to the ones in \cite{chassagneux2019rate}, where they derive the convergence for relaxed condition on the coefficient in the state-dependent case. The dependence on the path of $X$ will not introduce extra difficulty. To be more precise, adding the path-dependence of the generator $f$ and terminal condition $g$ on $(X_{\cdot\wedge t})_{0\le t\le T}$, we can update the discretized scheme in \cite{hamadene2010switching}[equation (3.2)] (or \cite{chassagneux2019rate}[equation (1.5)] and \cite{memin2008convergence}), as long as $f$ and $g$ are Lipschitz continuous (Assumption \eqref{main assumption}) for all the variables and $X$ itself is Markovian. Combining with the estimates in \eqref{path estiamte}, we get the similar estimates as below, 
\bea\label{RBSDE euler est}
\max_{i=0,\cdots, \tilde n-1}\hE|Y_{u_i}-Y^n_{u_i}|^2=\mathcal O(k\Delta t),\qq 
\hE[\sum_{i=0}^{\tilde n-1}\int_{u_i}^{u_{i+1}} |Z_t-Z_{u_i}^n |^2 dt ]=\mathcal O(\sqrt{k\Delta t}).
\eea
The rest of the proof follows closely to \cite{hure2020deep}[Theorem 4.4] after the following adjustment, 
\bea
\widehat{\mathcal U}_{u_i}(\sig(X^n)_{u_i})=\max[\mathcal U(\sig(X^n)_{u_i},\xi^*),g(u_i,X_{\cdot\wedge u_i}) ].
\eea 
First, we observe that 
\beaa
\widetilde Y_{u_i}^n-\widetilde{\mathcal V}_{u_i}
=\hE_{u_i}[Y_{u_{i+1}}^n-\widehat{\mathcal U}_{u_{i+1}}(\sig(X^n)_{ u_{i+1}}) ]+\Delta u_i( f(u_i,X_{\cdot\wedge u_i},\widetilde Y_{u_i}^n,Z^n_{u_i})-f(u_i,X^n_{\cdot\wedge u_i},\widetilde{\mathcal V}_{u_i},\overline{\widetilde Z}^n_{u_i}) ),
\eeaa
for $i=0,1\cdots, \tilde n-1$. Similar to \eqref{first inq}, we further get 
\begin{equation}
\label{first inq-variation}
\begin{split}
\hE|\widetilde Y_{u_i}^n-\widetilde{\mathcal V}_{u_i}|^2\le&  (1+\gamma \Delta u_i)\hE|\hE_{u_i}[Y^n_{u_{i+1}}-\widehat{\mathcal U}_{u_{i+1}}(\sig(X^n)_{ u_{i+1}})  ] |^2 \\
&+2\frac{L^2}{\gamma}(1+\gamma\Delta u_i) [\Delta u_i \hE|\widetilde Y_{u_i}^n-\widetilde{\mathcal V}_{u_i} |^2+\Delta u_i \hE|Z^n_{u_i}-\overline{\widetilde Z}_{u_i} |^2 ].
\end{split}
\end{equation}
Notice that \eqref{RBSDE-1} and \eqref{RBSDE-2} share the same $X_{\cdot\wedge u_i}$ dependence for $f$, we get the following similar estimates as in  \cite{hure2020deep}[Theorem 4.4],
\begin{equation}
\begin{split}
\Delta u_i\hE|Z^n_{u_i}-\overline{\widetilde Z_{u_i}}|^2 &\le 2d \Big[\hE|Y^n_{u_{i+1}}-\widehat{\mathcal U}_{u_{i+1}}(\sig(X^n)_{ u_{i+1}})|^2\\
&\qq\qq\qq\qq-\hE|\hE_{u_i}[Y^n_{u_{i+1}}-\widehat{\mathcal U}_{u_{i+1}}(\sig(X^n)_{ u_{i+1}}) ] |^2\Big].
\end{split}
\end{equation} 
Let $\gamma=4dL^2$, for $\Delta u_i$ small enough, we get 
\beaa
\hE|\widetilde Y^n_{u_i}-\widetilde{\mathcal V}_{u_i} |^2\le (1+Ck\Delta t)\hE|Y^n_{u_{i+1}}-\widehat{\mathcal U}_{u_{i+1}}(\sig(X^n)_{ u_{i+1}}) |^2.
\eeaa
Following the steps for deriving the estimates in \eqref{second inq}, we have
\begin{equation}
	\begin{split}
\label{variation last inq}
\hE|\widetilde Y^n_{u_i}-{\mathcal U}_{u_i}(\sig(X^n)_{u_i};\xi) |^2\le& (1+Ck\Delta t) \hE|Y_{u_{i+1}}^n-\widehat{\mathcal U}_{u_{i+1}}(\sig(X^n)_{u_{i+1}}) |^2 \\
&+C \tilde n \hE|\widetilde{\mathcal V}_{u_i}-{\mathcal U}_{u_i}(\sig(X^n)_{u_i};\xi) |^2.
	\end{split}
\end{equation}
Similar to \eqref{deterministic function} and \eqref{martingale rep}, applying the Martingale representation theorem for $\widehat{\mathcal U}_{u_{i+1}}(\sig(X^n)_{u_{i+1}})$, there exists an $\hR^d$-valued square integrable process $\widetilde Z$ such that 
\begin{equation}
\label{martingale rep-variation}
\widehat{\mathcal U}_{u_{i+1}}(\sig(X^n)_{u_{i+1}})
=\widetilde{\mathcal V}_{u_i}-f(u_i,X^n_{\cdot \wedge u_i}, \widetilde{\mathcal V}_{u_i}, \overline{\widetilde Z}_{u_i})\Delta u_i+\int_{u_i}^{u_{i+1}} \widetilde Z^{\ts}_tdW_t.
\end{equation}
Following \eqref{loss} in the previous case, in order to estimate the second term on the right hand side of \eqref{variation last inq}, we denote $\widehat L_{u_i}=\widetilde L_{u_i}+\hE[\int_{u_i}^{u_{i+1}}|\widetilde Z_t -\overline{\widetilde Z}_{u_i}|^2dt ]$, and 
\bea
&&\widetilde L_{u_i}(\theta):=\Delta u_i\hE|\overline{\widetilde Z}_{u_i}-\mathcal Z_{u_i}(\sig(X^n)_{u_i};\eta) |^2+\hE\Big|\widetilde V_{u_i}-\mathcal U_{u_i}(\sig(X^n)_{u_i};\xi) \nonumber\\
&&+[f(u_i,X^n_{\cdot\wedge u_i},\mathcal U_{u_i}(\sig(X^n)_{u_i};\xi),\mathcal Z_{u_i}(\sig(X^n)_{u_i};\eta))-f(u_i,X^n_{\cdot\wedge u_i},\widetilde{\mathcal V}_{u_i},\overline{\widetilde{Z}}_{u_i}) ] {\Delta} u_i\Big|^2.\nonumber
\eea
Following the steps in \eqref{error upper bound}, \eqref{interpolation}, \eqref{nn error euler} and applying the new error in \eqref{defn: NN error reflection}, recalling the fact $\widehat {\mathcal U}(\sig(X^n))_{u_i}=\max[\mathcal U(\sig(X^n)_{u_i};\xi^*);g(X^n_{\cdot\wedge u_i},u_i) ]$, $Y^n_{u_i}=\max[\widetilde Y^n_{u_i};g(X_{\cdot\wedge u_i},u_i) ]$, together with \eqref{variation last inq} and $|\max(a,c)-\max(b,c)|\le |a-b|$, we have 
\begin{equation}
\begin{split}
\hE|Y^n_{u_i}-\widehat{\mathcal U}_{u_i}(\sig(X^n)_{u_i}) |^2&\le (1+Ck\Delta t)\hE|Y^n_{u_{i+1}}-\widehat{\mathcal U}_{u_{i+1}}(\sig(X^n)_{u_{i+1}}) \|^2\\
&+C\tilde n( \varepsilon_{u_i}^{\cN,\tilde{ \mathsf v}}+\varepsilon_{u_i}^{\cN,\tilde{ \mathsf v}, \mathsf{Sig},m}+\Delta u_i (\varepsilon_{u_i}^{\cN,\tilde{ \mathsf z}}+\varepsilon_{u_i}^{\cN,\tilde{ \mathsf z}, \mathsf{Sig},m}) ).
\end{split}
\end{equation}
By induction and the estimate in \eqref{cummulated sig error}, we conclude
\bea
\label{est euler nn}
\max_{i=0,\cdots,\tilde n}\hE|Y^n_{u_i}-\widehat{\mathcal U}_{u_i}(\sig(X^n)_{u_i})|^2\le C \sum_{i=0}^{\tilde n-1}(\tilde n\varepsilon_{u_i}^{\cN,\tilde{ \mathsf v}}+ \varepsilon_{u_i}^{\cN,\tilde{ \mathsf z}}  )+\varepsilon^{\mathsf{Sig},m}.
\eea
Combining \eqref{RBSDE euler est} and \eqref{est euler nn}, we finish the proof. \qed 
\end{proof}

\section{Numerical Example}
\subsection{Amerasian Option}
We consider the following Amerasian option under the Black-Scholes model that involves  $d$ stocks $X_1, \dots, X_d$. The risk neural dynamics are given by  
\bea \label{Amerasian X}
dX_t^i = r X_t^i dt + \sigma_i X_t^{i} dW_t^i, \ X_0^i = x_0^i, \ i = 1, \dots, d,
\eea 
Here, $r$ represents the risk-free rate, $\sigma_i$ denotes the volatility, and $W^1, \dots, W^d$ are independent standard Brownian motions. The payoff of the basket Amerasian call option at a strike price of $K$ is defined as follows, where $(w_i)_{i=1, \dots, d}$ represents a vector of weights: 
\bea \label{defn: call}
g(X_{\cdot\wedge T}) = \left(\sum_{i=1}^d \frac{w_i}{T} \int_0^T X^i_t dt - K \right)^+.
\eea 
We have summarized the experimental results for Bermudan options in Table \ref{table:1}, with the following parameters:  $$X_0^i = 100, r = 5\%, \sigma_i = 0.15, w_i = \frac{1}{d}, T = 1, K = 100, \ \forall i \in \{1, \dots, d\}.$$
The European price is computed using Monte Carlo approximation to $e^{-rT}\mathbb{E}[g(X_{\cdot \wedge T})]$. Recall from Section \ref{sec-main}, we defined $n$ as the number of Euler scheme steps and $\tilde n$ as the number of segments for the signature layers, which corresponds to the number of exercise times in the Bermudan approximation. Bermudan options are a type of American option that permits early exercise, but only at predetermined dates. Using our method, we computed the Bermudan price with $\tilde{n} = 20$ and $n = 1000$. In contrast, the method from \cite{hure2020deep} uses $n = \tilde{n} = 20$, and the neural network takes all stock prices $X_t^i$ and their integrals $I_i(t)=\int_0^t X_s^i ds$ for $i = 1, \dots, d$ as inputs. By introducing the extra variables $I_i(t)$, for $i=1,\cdots, d$, we are able to transform the path-dependent problem to a state-dependent problem and apply the algorithm in \cite{hure2020deep}. 

We utilize a fully connected feedforward network featuring five hidden layers, each containing 16 neurons. The activation function for the hidden layers is tanh, while the output layer uses the identity function. To optimize our model, we use Adam Optimizer, which is implemented in TensorFlow, and utilize mini-batch with 100 trajectories for stochastic gradient descent. \footnote{The code can be found in the following URL link: \url{https://github.com/zhaoyu-zhang/sig_american}. The desktop we used in this study is equipped with an i7-8700 CPU. For all the examples in this paper, we generated 10,000 paths for the forward processes with a mini-batch of 100 paths.} 

We can compare our findings to those in \cite{dingpricing}, which put forth a formula for the price of the American option using a series expansion. However, this method cannot be used to price multi-dimensional American options. When using the same parameters but increased $\tilde{n}$ to 2024, the American option price in  \cite{dingpricing} was $4.6643$, with a $95\%$ confidence interval of $([4.5506, 4.7780])$. This result was slightly lower than the value of the European Asian option price computed using the Monte Carlo method in Table \ref{table:1}. Our results, which align with values from the Monte Carlo method of the European Asian option and the method in \cite{hure2020deep}, are shown in Table \ref{table:1}. Furthermore, our 95\% confidence interval from a sample of 100 runs is narrower, indicating greater stability in the price obtained from our method. Table \eqref{table:1} also reports the running time for obtaining a single value, which shows that the computation time is not exponential. Our algorithm is also approximately 4 times shorter than the method from \cite{hure2020deep}.

Our scheme can price high dimensional path-dependent American-type options. Applying the Jensen's inequality for $\hE[\cdot]$ and plugging in the parameters, we observe that $
\hE[e^{-rT}(\sum_{i=1}^d  \frac{1}{dT} \int_0^T X^i_t dt - K  )^+] \ge ( e^{-rT} \hE[ \frac{1}{d}\sum_{i=1}^d \frac{1}{T}\int_0^TX_t^idt-K])^+= 100e^{-r} (\frac{1}{d}\sum_{i=1}^d\int_0^1 \hE[$ $e^{(r-\frac{1}{2}\sigma_i^2)t+\sigma_iW_t^i}]dt-1 )^+=2.42$ for the European option price, which also provides a lower bound for the American option.  We can find an upper bound from the following observation:  For any stopping time $\tau\in[0, T]$, applying Jensen's inequality for the summation,  we have $e^{-r\tau}(\sum_{i=1}^d  \frac{1}{d\tau} \int_0^\tau X^i_t dt - K  )^+\le e^{-r\tau}\frac{1}{d}\sum_{i=1}^d(\frac{1}{\tau}\int_0^{\tau}X_t^idt-K )^+ $. After taking expectations and maximizing over $\tau$ and using the fact that the stocks are identically distributed, we see that $d=1$ case provides an upper bound.
Moreover, the value of the option is monotonically decreasing with respect $d$, which can be seen again using Jensen's inequality: for example, the value for the option $2d$ dimensions can be bounded by the $d$ dimensional counterpart using $e^{-r\tau}(\sum_{i=1}^{2d}\frac{1}{\tau(2d)}\int_0^\tau X_t^idt-K)^+\le \frac{1}{2}e^{-r\tau}(\sum_{i=1}^{d}\frac{1}{\tau d}\int_0^\tau X_t^idt-K)^++\frac{1}{2}e^{-r\tau}(\sum_{i=d+1}^{2d}\frac{1}{\tau d}\int_0^\tau X_t^idt-K)^+$. Furthermore, since the option price is decreasing in $d$, we can take the limit in $d$ first and use the Law of Large Numbers and Merton's no-early exercise theorem to conclude that the option price, in fact, converges as $d \to \infty$ to $X_0 e^{-rT}(e^{rT}/r-1/r-1)^+=2.42$. Our calculations in Table \ref{table:1} are in line with these theoretical observations.

The limitation of our signature algorithm is due to the fact that the dimension of the signature input grows exponentially in dimension. To reduce the limitation from the dimension of the signature is interesting on its own, we leave this for future studies.

 \begin{table}[h!]
	\begin{center}
		\begin{tabular}{ |c|c|c|c|c| } 
			\hline
			$d$  & Price  & CI  & Method & Time \\
			\hline
			\hline
			1 & 4.732 & -- & European Price & \\\hline
			 & 4.963 & [4.896, 5.03] & Our Method & 1631.85s \\
			 & 5.113 & [5.009, 5.217] & Method from \cite{hure2020deep} & 5457.59s\\\hline
			5 & 3.078 & -- & European Price & \\\hline
			 & 3.190 & [3.115, 3.266] & Our Method & 1927.55s\\
			 & 3.335 & [3.207, 3.462] & Method from\cite{hure2020deep}  & 5887.68s \\\hline
			10 & 2.701 & -- & European Price & \\\hline
			 & 2.914 & [2.844, 2.983] & Our Method &1947.64s\\
			 & 3.142 & [2.975, 3.309] & Method from \cite{hure2020deep} & 7636.96s \\\hline
			20 & 2.51 & -- & European Price & \\\hline
			  & 3.093 & [3.017, 3.168] & Our Method & 2287.57s\\
			  & 3.095 & [2.883, 3.308] & Method from \cite{hure2020deep} & 8887.61s\\ \hline
		\end{tabular}
	\end{center}
	\caption{Estimate of Bermudan option price ($Y_0$) from the average over 100 independent runs, and $95\%$ confidence interval (CI) of $Y_0$ are reported. }
	\label{table:1}
\end{table}

\subsection{Bermudan Option with moving average}
In this example, we study a truly path-dependent Bermudan option, which has been studied in \cite{BernardLelong2021} for one dimensional stock with moving average. We generalize this example to a more general and high dimensional truly path-dependent
 Bermudan put option and its payoff with a strike price $K$ is given by
\bea \label{defn: call path}
g(X_{\cdot\wedge T}) = \left(K - \sum_{i=1}^d  w_i I_i \right)^+,
\eea  
where $I_i = \left(\prod_{j=0}^{\tilde{n}-1} \frac{1}{t_{j+1} - t_{j}}\int_{t_j}^{t_{j+1}} X^i_t dt \right)^{1/\tilde{n}}$, and $t_j = jT/\tilde{n}$.
Our experiments on Bermudan options are summarized in Table \ref{table:2}, with the following parameters chosen:  $$X_0^i = 1, r = 5\%, \sigma_i = 0.15, w_i = \frac{1}{d}, T = 1, K = 1, n = 1000,  \ \forall i \in \{1, \dots, d\}.$$ 

 \begin{table}[h!]
	\begin{center}
		\begin{tabular}{ |c|c|c|c|c||c|c|c|c|c| } 
			\hline
			$d$ & $\tilde{n}$ & Price  & CI  & Time & $d$ & $\tilde{n}$ & Price  & CI & Time\\
			\hline
			\hline
			1 				 & 10 & 0.0277    & [0.0275, 0.0283] &  157.82s 
			& 10 & 10 & 0.0234  &  [0.0231, 0.0237] & 348.41s\\ \hline
			1				 & 20 & 0.0291 & [0.0272, 0.0311] & 386.10s 
			& 10 & 20 &  0.0247  &  [0.0242, 0.0253] & 852.41s\\  \hline
			1				 & 50 &  0.0367    & [0.0339, 0.0395] & 1391.24s
			 & 10 & 50 &  0.0265   & [0.0258, 0.0271] & 2779.21s\\  \hline
			\hline
			5 				 & 10 & 0.0258  & [0.0252, 0.0264] & 238.10s 
			& 20 & 10 & 0.0226 &[0.0217, 0.0223] &  653.64s\\ \hline
			5				 & 20 & 0.0273  &  [0.0266, 0.028] & 484.42s 
			& 20 & 20 &   0.0233  & [0.0228, 0.0237] & 1444.27s\\  \hline
			5				 & 50 &   0.0289  & [0.0277, 0.0301]& 1907.22s 
			& 20 & 50 & 0.0247 & [0.0239, 0.0255] & 4736.77s\\  \hline
		\end{tabular}
	\end{center}
	\caption{Estimate of Bermudan option price ($Y_0$) from the average over 100 independent runs, and $95\%$ confidence interval (CI) of $Y_0$ are reported.}
	\label{table:2}
\end{table}

\subsection{Shiryaev's optimal stopping problem}
In this section, we demonstrate the application of our algorithm to solve Shiryaev's optimal stopping problem \cite{bayraktarzhou2017}. 
\bea 
v^{\varepsilon}=\sup_{\tau\in \mathcal T_{0,T}} \mathbb E W_{(\tau-\varepsilon)^+},
\eea 
 where $\tau$ is a stopping time that takes values between $\varepsilon$ and $T$, a fixed time horizon. Here, $(W_t)_{0\le t\le T}$ represents Brownian motion, and $\varepsilon\in [0,T]$ is a constant delay parameter. We can obtain the solution to this problem by using the following path-dependent reflected BSDE (see \cite{bayraktarzhou2017}[Proposition 1]):
\bea 
&&W_{t-\varepsilon}\le Y_t=W_{T-\varepsilon} -\int_t^T Z_s dW_s+(K_T-K_t),\quad 0\le t\le T,\\
&&\int_0^T(Y_t-W_{t-\varepsilon}) dK_t = 0.
\eea 
It is important to note that this is a path-dependent reflected BSDE, which cannot be transformed into a state-dependent reflected BSDE. We apply our Algorithm \eqref{alg:the_alg} with $g(W_{\cdot \wedge t})=W_{t-\varepsilon}$. Both of the upper and lower bounds for the value $v^{(\varepsilon)}$ are established in
\cite{bayraktarzhou2017}.  When $\varepsilon$ falls within the range of $[0, T/2)$, $$\sqrt{\frac{2 \varepsilon }{\pi}}  <v^{(\varepsilon)}  \leq \sqrt{\frac{2 T }{\pi}}. $$ On the other hand, when $\varepsilon$ falls within the range of $[T/2, T]$, the value of $v^{(\varepsilon)}$ can be determined as $\frac{2(T-\varepsilon)}{\pi}$. Using $T=1$, Table \ref{table:3} displays the results of numerical experiments. The upper bound, which is $\sqrt{2 / \pi} = 0.798$, remains constant. However, the lower bounds vary for different values of $\varepsilon$, with values of 0.025, 0.08, and 0.252 for $\varepsilon = 0.001, 0.01,$ and $0.1$, respectively. Our numerical results fall within the established bounds.    Furthermore, as the number of segmentation increases, the value of $v^{(\varepsilon)}$ also increases due to the presence of more exercise dates.       When $\varepsilon$ lies within the range of $[0.5, 1]$, the exact value is given by $\sqrt{\frac{2(T-\varepsilon)}{\pi}}$.  For example, when $\varepsilon = 0.7$, our result approaches the explicit solution of 0.437 as the number of segmentation increases, as demonstrated in Figure \ref{pic:delay}.

  \begin{table}[h!]
	\begin{center}
		\begin{tabular}{ |c|c|c|c||c|c|} 
			\hline
			 \multicolumn{2}{|c|}{ } & \multicolumn{1}{c|}{$m=5$} & \multicolumn{1}{c||}{$m=10$} & lower bound & upper bound \\
			\hline
			$\varepsilon$  & $\tilde{n}$  & $v^{(\varepsilon)}$  & $v^{(\varepsilon)}$ & $\sqrt{2 \varepsilon /\pi}$   & $ \sqrt{2 T /\pi}$ \\\hline
			0.001 & 10 &  0.044   &  0.047 & 0.025 & 0.798   \\\hline
			 & 20 & 0.051     & 0.055  & 0.025 & 0.798 \\\hline
			 & 50 & 0.058      & 0.063  & 0.025 & 0.798 \\\hline
			0.01 & 10 & 0.117    & 0.131  & 0.08 & 0.798   \\\hline
			 & 20 & 0.165      & 0.161  & 0.08 & 0.798 \\\hline
			 & 50 & 0.219      & 0.192  & 0.08 & 0.798  \\\hline
			0.1 & 10 &  0.389     & 0.4  & 0.252 & 0.798 \\\hline
			 & 20 & 0.402     & 0.448  & 0.252 & 0.798 \\\hline
			 & 50 & 0.422       & 0.52  & 0.252 & 0.798   \\\hline
		\end{tabular}
	\end{center}
	\caption{Numerical Implementation of \cite{bayraktarzhou2017}.}
	\label{table:3}
\end{table}

\begin{figure}

\begin{center}
	\includegraphics[scale=0.6]{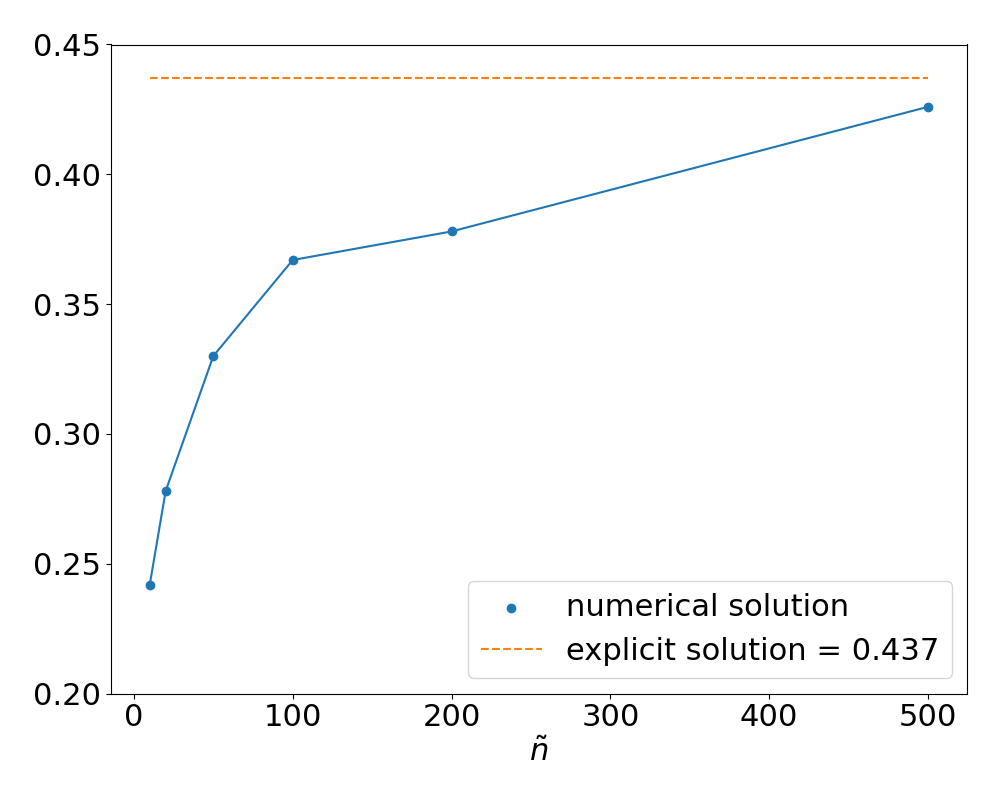}
	\caption{Numerical Implementation of \cite{bayraktarzhou2017} with delay = 0.7.}
	\label{pic:delay}
\end{center}
\end{figure}

\noindent\textbf{Acknowledgment} The authors would like to thank Xavier Warin for providing the code in \cite{hure2020deep}, and Huy\^en Pham for insightful discussions during his visit to U Michigan in April 2022.


\newpage
\begin{thebibliography}{10}

\bibitem{Aze} Robert Azencott. \emph{Formule de Taylor stochastique et d\'eveloppements asymptotiques d'int\'egrales de Feynman} In Az\'ema, Yor (Eds), S\'eminaire de probabilit\'es, XVI, LNM921 237--284,  Springer.

\bibitem{arribas2018derivatives}
Imanol~Perez Arribas.
\newblock Derivatives pricing using signature payoffs.
\newblock {\em arXiv preprint arXiv:1809.09466}, 2018.

\bibitem{auster2022jdoi}
Johan Auster, Ludovic Mathys, and Fabio Maeder.
\newblock {JDOI} variance reduction method and the pricing of {A}merican-style
  options.
\newblock {\em Quantitative Finance}, 22(4):639--656, 2022.


\bibitem{Baudoin}
Fabrice Baudoin.
\newblock \emph{An Introduction to the geometry of stochastic flows}, \newblock{Imperial College Press}, 140 pp, 2005.


\bibitem{Baudoin19}
Fabrice Baudoin, Eulalia Nualart, Cheng Ouyang. and Samy Tindel.
 \newblock \emph{On probability laws of solutions to differential systems driven by a fractional Brownian motion}.
 \newblock  Ann. Probab. 44(2016), no.4, 2554--2590.

\bibitem{BaudoinFengOuyang}
Fabrice Baudoin, Qi Feng, and Cheng Ouyang.
\newblock {\em Density of the signature process of fBm}.
\newblock {\em Trans. Amer. Math. Soc.} 373 (2020), no. 12, 8583-8610. 


\bibitem{barraquand1996pricing}
Jerome Barraquand and Thierry Pudet.
\newblock Pricing of American path-dependent contingent claims.
\newblock {\em Mathematical Finance}, 6(1):17--51, 1996.

\bibitem{bayraktarzhou2017}
Erhan Bayraktar and Zhou Zhou.
\newblock On an Optimal Stopping Problem of an Insider.
\newblock{\em Theory of Probability \& Its Applications}, 61(2017), no. 1, 129-133.

\bibitem{bayer2021optimal}
Christian Bayer, Paul Hager, Sebastian Riedel, and John Schoenmakers.
\newblock Optimal stopping with signatures.
\newblock {\em arXiv preprint arXiv:2105.00778}, 2021.

\bibitem{bayer2022pricing}
Christian Bayer, Jinniao Qiu, and Yao Yao.
\newblock Pricing options under rough volatility with backward {SPDE}s.
\newblock {\em SIAM Journal on Financial Mathematics}, 13(1):179--212, 2022.

\bibitem{bayraktar2022neural}
Erhan Bayraktar, Asaf Cohen, and April Nellis.
\newblock A neural network approach to high-dimensional optimal switching
  problems with jumps in energy markets.
\newblock {\em arXiv preprint arXiv:2210.03045}, 2022.


\bibitem{becker2019}
Sebastian Becker, Patrick Cheridito, and Arnulf Jentzen.
\newblock Deep optimal stopping.
\newblock {\em The Journal of Machine Learning Research.} 20.1 (2019): 2712-2736.

\bibitem{becker2020}
Sebastian Becker, Patrick Cheridito, and Arnulf Jentzen.
\newblock Pricing and Hedging American-Style Options with Deep Learning
\newblock{\em Journal of Risk Financial Management. 2020, 13, 158. https://doi.org/10.3390/jrfm13070158.}

\bibitem{becker2021solving}
Sebastian Becker, Patrick Cheridito, Arnulf Jentzen, and Timo Welti.
\newblock Solving high-dimensional optimal stopping problems using deep
  learning.
\newblock {\em European Journal of Applied Mathematics}, 32(3):470--514, 2021.


\bibitem{benezet2022switching}
Cyril Benezet, Jean-Francois Chassagneux, and Adrien Richou.
\newblock Switching problems with controlled randomisation and associated
  obliquely reflected {BSDE}s.
\newblock {\em Stochastic Processes and their Applications}, 144:23--71, 2022.

\bibitem{Cast} Fabienne Castell. Asymptotic expansion of stochastic flows, Probab.  Theory Relat. Fields (1993), \textbf{96}, 225--239, .


\bibitem{carmona2008pricing}
Ren{\'e} Carmona and Michael Ludkovski.
\newblock Pricing asset scheduling flexibility using optimal switching.
\newblock {\em Applied Mathematical Finance}, 15(5-6):405--447, 2008.

\bibitem{chassagneux2019rate}
Jean-Francois Chassagneux and Adrien Richou.
\newblock Rate of convergence for the discrete-time approximation of reflected
  {BSDE}s arising in switching problems.
\newblock {\em Stochastic Processes and their Applications},
  129(11):4597--4637, 2019.

\bibitem{dingpricing}
Kailin Ding, Zhenyu Cui, and Xiaoguang Yang.
\newblock Pricing arithmetic {A}sian and {A}merasian options: a diffusion
  operator integral expansion approach.
\newblock {\em Journal of Futures Markets, 2022, forthcoming.}

\bibitem{el1997reflected}
Nicole El~Karoui, Christophe Kapoudjian, Etienne Pardoux, Shige Peng, and
  Marie-Claire Quenez.
\newblock Reflected solutions of backward {SDE}'s, and related obstacle
  problems for {PDE}'s.
\newblock {\em the Annals of Probability}, 25(2):702--737, 1997.



\bibitem{fengzhang}
Qi~Feng and Xuejing Zhang.
\newblock Taylor Expansions and Castell Estimates for Solutions of Stochastic Differential Equations Driven by Rough Paths. 
\newblock {\em Journal of Stochastic Analysis: Vol. 1 : No. 2 , Article 4.}, 2020.


\bibitem{feng2021deep}
Qi~Feng, Man Luo, and Zhaoyu Zhang.
\newblock Deep signature {FBSDE} algorithm.
\newblock {\em Numerical Algebra, Control and Optimization}, 2022.

\bibitem{friz2010} 
Peter K. Friz and Nicolas B. Victoir
\newblock{\em Multidimensional stochastic processes as rough paths: theory and applications}
\newblock {\em Cambridge University Press}, vol 120, 2010.


\bibitem{frizhairer}
Peter K Friz and Martin Hairer.
\newblock {\em A course on rough paths},
\newblock{ \em Springer}, 2020.




\bibitem{funahashi1993approximation}
Ken-ichi Funahashi and Yuichi Nakamura.
\newblock Approximation of dynamical systems by continuous time recurrent
  neural networks.
\newblock {\em Neural networks}, 6(6):801--806, 1993.

\bibitem{gao2022convergence}
Chengfan Gao, Siping Gao, Ruimeng Hu, and Zimu Zhu.
\newblock Convergence of the backward deep {BSDE} method with applications to
  optimal stopping problems.
\newblock {\em arXiv preprint arXiv:2210.04118}, 2022.

\bibitem{gonon2022deep}
Lukas Gonon.
\newblock Deep neural network expressivity for optimal stopping problems.
\newblock {\em arXiv preprint arXiv:2210.10443}, 2022.

\bibitem{hamadene2010switching}
Said Hamadene and Jianfeng Zhang.
\newblock Switching problem and related system of reflected backward {SDE}s.
\newblock {\em Stochastic Processes and their applications}, 120(4):403--426,
  2010.

\bibitem{han2018solving}
Jiequn Han, Arnulf Jentzen, and Weinan E.
\newblock Solving high-dimensional partial differential equations using deep
  learning.
\newblock {\em Proceedings of the National Academy of Sciences},
  115(34):8505--8510, 2018.

\bibitem{han2020convergence}
Jiequn Han and Jihao Long.
\newblock Convergence of the deep {BSDE} method for coupled {FBSDE}s.
\newblock {\em Probability, Uncertainty and Quantitative Risk}, 5(1):1--33,
  2020.

\bibitem{hansen2000analytical}
Asbjorn~T Hansen and Peter~Lochte Jorgensen.
\newblock Analytical valuation of {A}merican-style {A}sian options.
\newblock {\em Management Science}, 46(8):1116--1136, 2000.

\bibitem{herrera2021optimal}
Calypso Herrera, Florian Krach, Pierre Ruyssen, and Josef Teichmann.
\newblock Optimal stopping via randomized neural networks.
\newblock {\em arXiv preprint arXiv:2104.13669}, 2021.

\bibitem{hu2022mean}
Ying Hu, Remi Moreau, and Falei Wang.
\newblock Mean-field reflected {BSDE}s: the general lipschitz case.
\newblock {\em arXiv preprint arXiv:2201.10359}, 2022.

\bibitem{hure2020deep}
C{\^o}me Hur{\'e}, Huy{\^e}n Pham, and Xavier Warin.
\newblock Deep backward schemes for high-dimensional nonlinear {PDE}s.
\newblock {\em Mathematics of Computation}, 89(324):1547--1579, 2020.

%


\bibitem{kidger2019deep}
Patrick Kidger, Patric Bonnier, Imanol~Perez Arribas, Cristopher Salvi, and
  Terry Lyons.
\newblock Deep signature transforms.
\newblock In {\em Advances in Neural Information Processing Systems}, pages
  3105--3115, 2019.
  
  
  
  \bibitem{kalsi2020}
Jasdeep  Kalsi,  Terry  Lyons and Imanol Perez Arribas
\newblock Optimal Execution with Rough Path Signatures\
\newblock {\em SIAM J. Financial Math.}, Vol. 11, No. 2, pages 470-493, 2020.
  
  

\bibitem{lyonsqian}
Terry Lyons and  Zhongmin Qian.
\newblock{\em System control and rough paths.}
\newblock{Oxford University Press,} 2002.



\bibitem{li2014reflected}
Juan Li.
\newblock Reflected mean-field backward stochastic differential equations.
  approximation and associated nonlinear {PDE}s.
\newblock {\em Journal of Mathematical Analysis and Applications},
  413(1):47--68, 2014.
  
\bibitem{lelong2019}
J\'er\^ome Lelong.
\newblock Pricing Path-Dependent Bermudan Options Using Wiener Chaos Expansion: An Embarrassingly Parallel Approach 
\newblock {\em Journal of Computational Finance}, Vol. 24, No. 2, 2019.

\bibitem{BernardLelong2021}
Bernard  Lapeyre and J\'er\^ome Lelong. 
\newblock Neural network regression for Bermudan option pricing
\newblock {\em Monte Carlo Methods and Applications}, vol. 27, no. 3, 2021, pp. 227-247.


\bibitem{longstaff2001valuing}
Francis~A Longstaff and Eduardo~S Schwartz.
\newblock Valuing american options by simulation: a simple least-squares
  approach.
\newblock {\em The review of financial studies}, 14(1):113--147, 2001.

\bibitem{ma2005representations}
Jin Ma and Jianfeng Zhang.
\newblock Representations and regularities for solutions to {BSDE}s with
  reflections.
\newblock {\em Stochastic processes and their applications}, 115(4):539--569,
  2005.

\bibitem{memin2008convergence}
Jean M{\'e}min, Shi-ge Peng, and Ming-Yu Xu.
\newblock Convergence of solutions of discrete reflected backward {SDE}s and
  simulations.
\newblock {\em Acta Mathematicae Applicatae Sinica, English Series},
  24(1):1--18, 2008.


\bibitem{platen2010numerical}
Eckhard Platen and Nicola Bruti-Liberati.
\newblock {\em Numerical solution of stochastic differential equations with
  jumps in finance}, volume~64.
\newblock Springer Science \& Business Media, 2010.






\bibitem{Sabate2020}
Marc Sabate-Vidales, Daviv Siska, and Lukasz Szpruch.
\newblock Solving path dependent PDEs with LSTM networks and path signatures.
\newblock{ \em arXiv preprint, 	arXiv:2011.10630, 2020.}

\bibitem{schmit2020}
Johannes Schmidt-Hieber.
\newblock{\em Nonparametric regression using deep neural networks with ReLU activation function.}
\newblock Ann. Statist. 48(4): 1875-1897 (August 2020). DOI: 10.1214/19-AOS1875.

\bibitem{wang2018deep}
Haojie Wang, Han Chen, Agus Sudjianto, Richard Liu, and Qi~Shen.
\newblock Deep learning-based {BSDE} solver for {LIBOR} market model with
  application to {Bermudan} swaption pricing and hedging.
\newblock {\em arXiv preprint arXiv:1807.06622}, 2018.

\bibitem{wang2022deep}
Yutian Wang and Yuan-Hua Ni.
\newblock Deep {BSDE}-{ML} learning and its application to model-free optimal
  control.
\newblock {\em arXiv preprint arXiv:2201.01318}, 2022.

\bibitem{zhang2004numerical}
Jianfeng Zhang.
\newblock A numerical scheme for {BSDE}s.
\newblock {\em The Annals of Applied Probability}, 14(1):459--488, 2004.

\end{thebibliography}
\end{document}